\documentclass[lettersize,journal]{IEEEtran}
\usepackage{array}
\usepackage[caption=true,font=footnotesize,labelfont=rm,textfont=rm]{subfig} 
\usepackage{textcomp}
\usepackage{stfloats}
\usepackage{url}
\usepackage{verbatim}
\usepackage{graphicx}
\def\BibTeX{{\rm B\kern-.05em{\sc i\kern-.025em b}\kern-.08em
    T\kern-.1667em\lower.7ex\hbox{E}\kern-.125emX}}
\usepackage{balance}

\usepackage{amssymb,amsfonts} 
\usepackage[tbtags]{amsmath}

\usepackage[ruled,linesnumbered]{algorithm2e} 

\usepackage{color} 
\usepackage{booktabs} 
\usepackage[thmmarks,amsmath]{ntheorem} 
\usepackage{setspace} 
\usepackage{hyperref}
\usepackage{cite}
\usepackage{soul}
\soulregister{\cite}7 
{
    \theoremstyle{nonumberplain} 
    \theoremheaderfont{\bfseries\itshape} 
    \theorembodyfont{\normalfont} 
    \theoremseparator{:} 
    \theoremsymbol{\ensuremath{\blacksquare}} 
    \newtheorem{proof}{Proof}
}
{
    \theoremstyle{plain} 
    \theoremheaderfont{\bfseries\itshape} 
    \theorembodyfont{\normalfont} 
    \theoremseparator{:} 
    \theoremsymbol{} 
    
    \newtheorem{definition}{Definition}
    \newtheorem{lemma}{Lemma}
    \newtheorem{theorem}{Theorem}
}

\newcommand{\algref}[1]{\textit{\textbf{Algorithm \ref{#1}}}}
\newcommand{\theoref}[1]{\textit{\textbf{Theorem \ref{#1}}}}
\newcommand{\lemmref}[1]{\textit{\textbf{Lemma \ref{#1}}}}
\newcommand{\defref}[1]{\textit{\textbf{Definition \ref{#1}}}}
\newcommand{\figref}[1]{Fig. \ref{#1}}

\begin{document}
\title{Flexible Base Station Sleeping and Resource Allocation for Green Uplink Fully-Decoupled RAN}

\author{Yu Sun,~\IEEEmembership{Member,~IEEE}, Haibo Zhou,~\IEEEmembership{Senior Member,~IEEE}, Kai Yu,~\IEEEmembership{Member,~IEEE}, Yunting Xu,~\IEEEmembership{Member,~IEEE}, Bo Qian,~\IEEEmembership{Member,~IEEE}, and Lin X. Cai,~\IEEEmembership{Senior Member,~IEEE}
	\thanks{This work is supported in part by the National Natural Science Foundation of China under Grant 62271244 and the Jiangsu Province Innovation and Entrepreneurship Team Project under Grant JSSCTD202202. (\textit{Corresponding author: Haibo Zhou.})\\
		Y. Sun and H. Zhou are with the School of Electronic Science and Engineering, Nanjing University, Nanjing, China, 210023 (e-mail: \href{yusun@smail.nju.edu.cn}{yusun@smail.nju.edu.cn}, \href{mailto:haibozhou@nju.edu.cn}{haibozhou@nju.edu.cn}).\\
		K. Yu is with the School of Electrical Engineering and Computer Science, KTH Royal Institute of Technology, 114 28 Stockholm, Sweden (e-mail: \href{kayu@kth.se}{kayu@kth.se}).\\
		Y. Xu is with the College of Computing and Data Science, Nanyang Technological University, Singapore, 639798 (e-mail: \href{yunting.xu@ntu.edu.sg}{yunting.xu@ntu.edu.sg}).\\
		B. Qian is with the Information Systems Architecture Science Research Division, National Institute of Informatics, Tokyo 101-8430, Japan (e-mail: \href{boqian@ieee.org}{boqian@ieee.org}).\\
		L. X. Cai is with the Department of Electrical and Computer Engineering, Illinois Institute of Technology, Chicago, IL 60616, USA (e-mail: \href{lincai@iit.edu}{lincai@iit.edu}).}
}

\maketitle

\begin{abstract}
	The fully-decoupled radio access network (FD-RAN) is an innovative architecture designed for next-generation mobile communication networks, featuring decoupled control and data planes as well as separated uplink and downlink transmissions.
	To further enhance energy efficiency, this paper explores a green approach to FD-RAN by incorporating adaptive base station (BS) sleeping and resource allocation.
	First, we introduce a holistic power consumption model and formulate a  energy efficiency maximization problem for FD-RAN, involving joint optimization of user equipment (UE) association, BS sleeping, and power control.
	Subsequently, the optimization problem is decomposed into two subproblems. The first subproblem, involving UE power control, is solved using a successive lower-bound maximization approach based on Dinkelbach's algorithm. The second subproblem, addressing UE association and BS sleeping, is tackled via a modified, low-complexity many-to-many swap matching algorithm.
	Extensive simulation results demonstrate the superior effectiveness of FD-RAN with our proposed algorithms, revealing the sources of energy efficiency gains.
\end{abstract}


\section{Introduction}
\IEEEPARstart{T}{he} sixth-generation network is expected to meet the growing demand for wireless communications through denser deployments, cloud and edge computing, and the integration of artificial intelligence \cite{zhangHandoverfreeMobilityManagement2024,10679152,LAIM_CP}. However, these advancements result in a significant increase in energy consumption \cite{chenEvolutionRANArchitectures2024,10531073}.
Notably, Vodafone reports that network base station (BS) sites contribute to nearly 73\% of the total energy consumption, and this trend is on the rise \cite{2022ESGAddendum2022}. Thus, reducing the energy consumption of BSs is crucial for developing environmentally friendly and sustainable wireless communication networks.
One potential approach is the adoption of sleep mechanisms for BSs, which allow underutilized BSs to enter sleep mode and offload their traffic to nearby BSs cooperatively, leading to significant energy savings \cite{liuDNNPartitioningTask2024a}. However, implementing BS sleeping may introduce challenges, particularly the creation of coverage holes when a BS enters a sleep mode and temporarily stops serving wireless users. These coverage holes can disrupt seamless service and degrade the overall network performance and user experience.

In a traditional cellular network where the uplink (UL) and downlink (DL) are tightly coupled, legacy BSs can only sleep when both are idle. This limitation reduces the effectiveness of BS sleeping, preventing energy savings during periods when only UL or DL is inactive.
In contrast, the fully-decoupled radio access network (FD-RAN) \cite{chenEvolutionRANArchitectures2024,xuFullyDecoupledRANFeedbackFree2025a}, a novel and disruptive architecture for next-generation mobile communication networks, is poised to address these challenges.
In an FD-RAN, BSs are decoupled into control BSs (CBSs), uplink BSs (UBSs), and downlink BSs (DBSs), allowing for the full decoupling of the control and data planes, as well as uplink and downlink transmissions. This decoupling improves BS sleeping by enabling BSs to independently manage their sleep modes for control and data transmissions, enhancing energy efficiency by allowing them to sleep when only one plane or direction is active.
Under this architecture, CBSs remains active and provide always-on and ubiquitous coverage, while UBSs and DBSs can dynamically enter sleep mode based on the traffic demands. This fundamentally resolves the coverage hole issue caused by BS sleeping in traditional networks and allows UBSs and DBSs to sleep independently, enabling an optimal sleep strategy.
Furthermore, the inherent multi-connectivity mode combined with adaptive resource allocation is expected to further support BS sleeping and improve energy efficiency \cite{sunJointUserAssociation2023,zhaoFullydecoupledRadioAccess2023,xueCooperativeDeepReinforcement2024}.

Although BS sleeping has been explored in traditional architectures \cite{linDatadrivenBaseStation2024a,masoudiDigitalTwinAssisted2023,zhouJointUserAssociation2024,salahdineSurveySleepMode2021,liuDeepNapDataDrivenBase2018}, little research has been done in the context of FD-RAN. Generally, several unique challenges arise when optimizing BS sleeping in FD-RAN architectures:
1) The holistic modeling of energy consumption in an FD-RAN, which is crucial for energy efficiency studies, is yet to be fully developed;
2) The intrinsic multi-connectivity in an FD-RAN significantly increases the complexity of the NP-hard BS sleeping problem, expanding it from $2^M$ to $2^{MK}$, where $M$ and $K$ denote the number of BSs and UEs, respectively;
3) Realistic power consumption models and multi-connectivity considerations further complicate the non-convex energy efficiency maximization problem.

In this work, we study adaptive UBS sleeping and resource allocation in an uplink FD-RAN to maximize the overall network energy efficiency.
To strike a balance between accuracy and tractability in assessing energy efficiency, we first propose a holistic power consumption model for an FD-RAN. Based on the developed power model, we formulate an energy efficiency maximization problem as a mixed-integer nonlinear programming (MINLP) problem. The objective of the problem is to maximize network energy efficiency while ensuring user equipment (UE) quality of service (QoS) by optimizing the UE association, UBS sleeping schedule, and transmission power control. The formulated MINLP problem is NP-hard and can be decomposed into two subproblems. The firsty subproblem involving UE power control in continuous space is solved using the successive lower-bound maximization based Dinkelbach's (SLMDB) algorithm \cite{dinkelbachNonlinearFractionalProgramming1967}.
The second subproblem, which involves UE association and BS sleeping, is addressed using a modified many-to-many swap matching algorithm (TriMSM) with a low-complexity implementation, striking a favorable balance between performance and computational efficiency.

The major contributions of this paper are summarized as follows:
\begin{itemize}
	\item We first develop a power consumption model that captures the key components of the FD-RAN infrastructure and UEs. Based on this model, we formulate an energy efficiency maximization problem, considering the multi-connectivity and power control, while ensuring QoS for UEs.
	\item To solve the complex optimization problem, we decompose it into two subproblems: the power control subproblem and the UE association with UBS sleeping subproblem.
	\item We propose the SLMDB algorithm to address the continuous yet non-convex power control subproblem. The subproblem is reformulated as a sequence of lower-bounded concave-convex fractional programs with guaranteed global convergence, which are then solved using the Dinkelbach's algorithm.
	\item We propose the TriMSM algorithm and leverage its low-complexity realizations to address the nonlinear integer UE association and UBS sleeping subproblem. A modified many-to-many swap matching algorithm is presented, and three alternative low-complexity power control algorithms are proposed to ensure superior performance while significantly reducing overall computational complexity.
	\item Extensive simulations validate the effectiveness of the proposed algorithms for a FD-RAN, and reveal the underlying sources of energy efficiency improvements. Specifically, the proposed approach achieves at least an 18.9\% gain in energy efficiency compared to conventional architectures, and at least a 6.6\% improvement over some baseline algorithms in the literature.
\end{itemize}

The remainder of this paper is organized as follows. Section~\ref{sec:related_works} reviews the related works in the literature. Section~\ref{sec:system_model} presents the network model, the data communication model, and the power consumption model for FD-RAN. The energy efficiency maximization problem is formulated in Section~\ref{sec:problem_formulation}, along with the overall solution framework. Detailed algorithms for the subproblems are presented in Section~\ref{sec:SLMDB} and Section~\ref{sec:joint_alg}, respectively. Extensive simulation results are provided in Section~\ref{sec:simulation_results}, followed by concluding remarks in Section~\ref{sec:conclusion}.

\section{Related Works} \label{sec:related_works}

The construction of accurate power consumption models is crucial for evaluating the energy efficiency of communication systems, and a significant body of work has been dedicated to this area~\cite{debaillieFlexibleFutureProofPower2015,auerHowMuchEnergy2011,dessetFlexiblePowerModeling2012,fioraniModelingEnergyPerformance2016,basharEnergyEfficiencyCellFree2019,vanchienJointPowerAllocation2020}.
However, these models cannot be directly applied to FD-RAN, and existing studies only focus on a part of network power consumption, lacking comprehensiveness.
For instance, Auer et al. \cite{auerHowMuchEnergy2011} proposed several power models for BSs, covering macro, micro, pico, and femto cells; however, their models are limited to BSs and exclude other network components. Similarly, Bashar et al. \cite{basharEnergyEfficiencyCellFree2019} considered the power consumption of BSs, UEs, and backhauls, but their study only applies to distributed networks, excluding edge clouds. Fiorani et al. \cite{fioraniModelingEnergyPerformance2016} developed models for the radio network and optical transport network, incorporating centralized control and processing, but simplified other components and omitted UEs. Moreover, Vanchien et al. \cite{vanchienJointPowerAllocation2020} explored power consumption in fronthauls, yet their model still fails to account for UEs and edge clouds. Debaillie et al. \cite{debaillieFlexibleFutureProofPower2015} and Desset et al. \cite{dessetFlexiblePowerModeling2012} investigated the impact of BS sleeping on power consumption, but their models are only applicable to legacy BSs and exclude other components.
Additionally, some power models \cite{basharEnergyEfficiencyCellFree2019,vanchienJointPowerAllocation2020} prioritize simplicity for ease of implementation, sacrificing accuracy in the process. In contrast, models based on real data \cite{debaillieFlexibleFutureProofPower2015,dessetFlexiblePowerModeling2012} offer higher accuracy but are often more complex and difficult to solve.
Therefore, to assess energy efficiency effectively, a comprehensive approach is required that balances reliable data with manageable complexity, ensuring both accuracy and tractability.
Our proposed model for power consumption incorporates key components of the FD-RAN infrastructure, including BSs, fronthauls, edge clouds, and UEs, as well as being specifically designed for FD-RANs.

BS sleeping, a potential method for substantial energy savings, encounters implementation challenges, with coverage holes emerging as a significant concern during sleeping \cite{wuEnergyEfficientBaseStationsSleepMode2015}. Various approaches have been proposed to tackle this challenge.
Lin et al. developed a spatio-temporal traffic prediction model aimed at capturing traffic characteristics to efficiently manage arriving UE traffic in BS sleeping schemes \cite{linDatadrivenBaseStation2024a}.
To optimize the timing of BS sleeping, Masoudi et al. \cite{masoudiDigitalTwinAssisted2023} utilized a digital twin model to encapsulate the dynamic system behavior and estimate risks in advance.
Zhou et al. introduced a BS sleeping scheme ensuring continuous coverage by macro BSs while allowing small BSs to dynamically sleep for energy conservation \cite{zhouJointUserAssociation2024}.
Additionally, recent studies explore promising techniques like UE association, self-organizing networks, and cell zooming \cite{salahdineSurveySleepMode2021}.
Nevertheless, heterogeneous deployment introduces additional costs and energy consumption due to the need for diverse infrastructure, integration challenges, and the varying power requirements of different network components. Other techniques, while mitigating the adverse effects of coverage holes, fail to address the problem fundamentally. Furthermore, the coupled uplink and downlink transmission within these works hinders optimal BS sleeping.

Resource allocation plays a pivotal role in improving the efficiency and reliability of wireless communication networks. To achieve optimal resource allocation, the problem is typically modeled as a complex optimization challenge. In an effort to tackle this complexity, Ma et al. \cite{maUAVLEOIntegratedBackbone2021} utilized the block coordinate descent (BCD) based algorithm to decompose the problem into sub-problems, and alternatively solve them until convergence.
Within the domain of energy efficiency, problems are often formulated as fractional programming. Shen et al. \cite{shenFractionalProgrammingCommunication2018} highlighted a specific subset of these issues, termed concave-convex fractional problems, which can be addressed effectively. For more general cases, the need for transformations or approximations depends on the specific nature of the problem under consideration. For instance, Huang et al. \cite{huangEnergyEfficientIntegratedSensing2023} obtained a more solvable form of the fractional problem by introducing a new auxiliary variable.
Ma et al. \cite{maUAVLEOIntegratedBackbone2021} employed the Lagrange partial relaxation method to transform integer variables into continuous ones, formulating the dual problem. Subsequently, they restored the relaxed variables back to integers, effectively addressing the problem.
Guo et al. \cite{guoJointPowerUser2022} utilized the generalized benders decomposition method, iteratively solving the primal and master problems to handle the integer variables.
Qian et al. \cite{qianEnablingFullyDecoupledRadio2023} introduced a UE-BS-subchannel matching game using many-to-many matching, proven to converge towards a stable matching.
Di et al. \cite{diSubChannelAssignmentPower2016} treated users and sub-channels as players, formulating the sub-channel assignment problem as a swap many-to-many matching game that converges to a two-sided exchange-stable matching.
However, these methods face limitations in effectively managing and directly resolving the intricate non-convex problem in FD-RAN. Moreover, existing matching algorithms lack detailed descriptions on handling QoS constraints and fall short in terms of low-complexity implementations.

\section{System Model} \label{sec:system_model}
\begin{table*}[!t]
	\centering
	\caption{Key Notations}
	\label{tab:notations}
	\renewcommand\arraystretch{1.2}
	\resizebox{.8\linewidth}{!}{
		\begin{tabular}{ll|ll}
			\toprule
			\hline
			\textbf{Notation}           & \textbf{Description}                            & \textbf{Notation}   & \textbf{Description}                           \\ \hline
			$M$                         & The number of UBSs                              & $S_{m,k}$           & The association between UE $k$ and UBS $m$     \\
			$K$                         & The number of UEs                               & $A_m$               & The operating status of UBS $m$                \\
			$N$                         & The number of antennas equipped with UBS        & $P_k$               & Transmit power of UE $k$                       \\
			$\mathcal{M}_k$             & Set of UBSs serving for UE $k$                  & $\mathbf{h}_{m,k}$  & Channel response between UE $k$ and UBS $m$    \\
			$\mathcal{K}_m$             & Set of UEs served by UBS $m$                    & $\mathbf{R}_{m, k}$ & Spatial correlation between UE $k$ and UBS $m$ \\
			$L$                         & Maximum number of UBSs allowed to serve each UE & $P_{\mathrm{N}}$    & Holistic network power consumption             \\
			$\mathrm{DS}_k$             & Desired signal component for UE $k$             & $\mathrm{EE}$       & Energy efficiency                              \\
			$\mathrm{IS}_{k, k^\prime}$ & Interference from UE $k^\prime$ to UE $k$       & $B$                 & Allocated bandwidth                            \\
			$\mathrm{NS}_{m,k}$         & Noise power at UBS $m$ for signal from UE $k$   & $R_k$               & Uplink rate of UE $k$                          \\ \hline
			\bottomrule
		\end{tabular}
	}
\end{table*}
\begin{figure}[!t]
	\captionsetup{width=\linewidth}
	\centerline{\includegraphics[width=\linewidth]{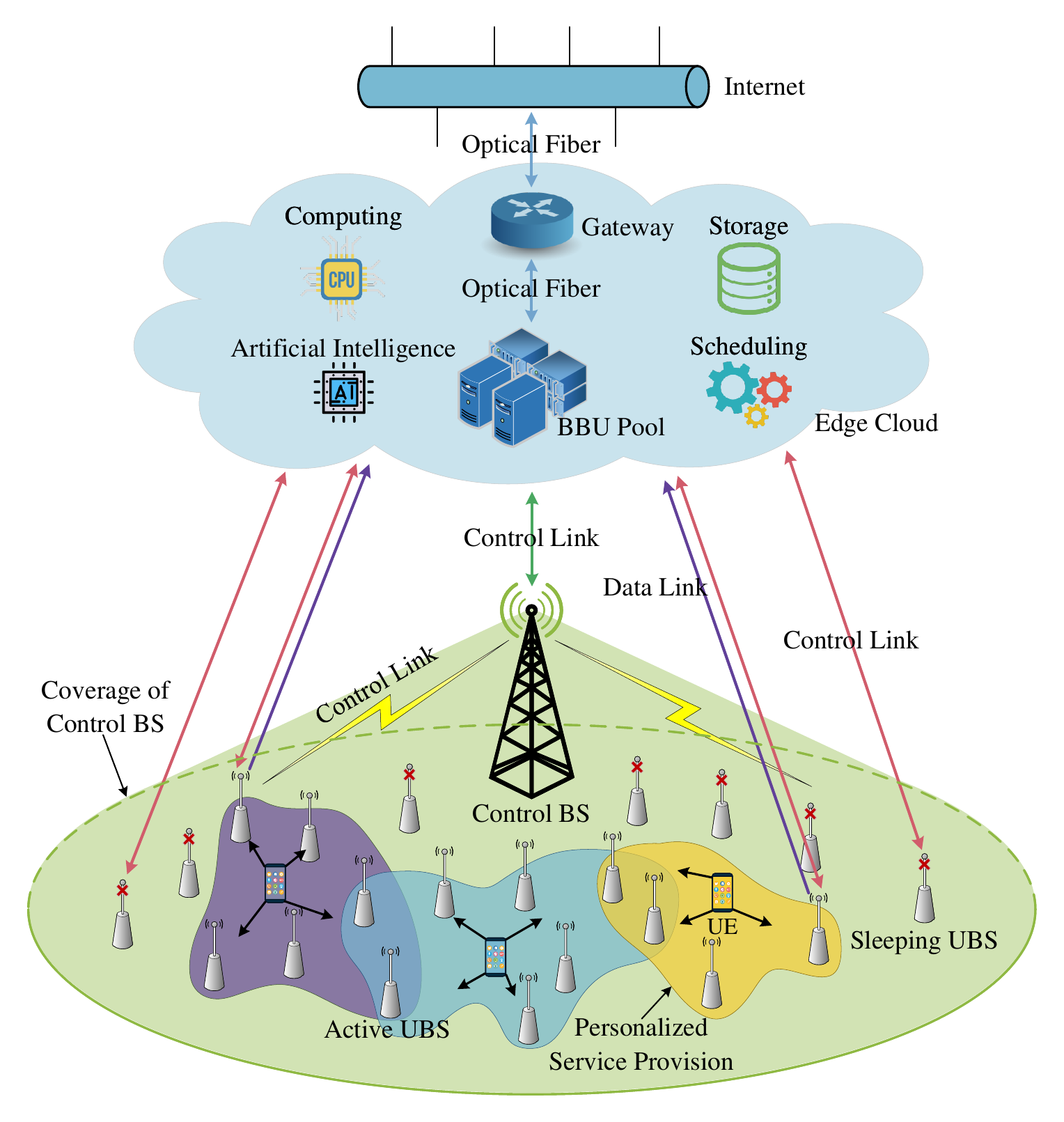}}
	\caption{UE association and UBS sleeping in green uplink FD-RAN.}
	\label{fig:system_model}
\end{figure}

\subsection{Network Model} \label{sec:network_model}
We consider an uplink FD-RAN scenario depicted in \figref{fig:system_model}, where a CBS is responsible for the control of data plane, while $M$ UBSs handle the data plane functions. Additionally, an edge cloud is deployed to provide centralized control and processing.
In the control plane, there are $K$ single-antenna UEs and $M$ UBSs equipped with $N$ antennas underlaid within the coverage of the CBS.
Let $\mathcal{M} = \left\{ 1, \cdots, M \right\}$ and $\mathcal{K} = \left\{ 1, \cdots, K \right\}$ denote the sets of UBSs and UEs, respectively.
In the data plane, multiple cooperative UBSs form a set $\mathcal{M}_k \subset \mathcal{M}$ to provide services for UE $k$. Correspondingly, the UE set $\mathcal{K}_m \subset \mathcal{K}$ represents the set of UEs served by UBS $m$. After receiving data from UEs, UBSs forward the data to the edge cloud via a wired fronthaul. Notice that there exists no active data link between an sleeping UBS and the edge cloud.
The associations between UEs and UBSs are represented by matrix $\mathbf{S} = \left( S_{m,k} \right)_{M\times K}$, where the binary variable $S_{m,k}=1$ indicates that UBS $m$ serves UE $k$, and $S_{m,k} = 0$ otherwise.
Specifically, if $\sum_{k \in \mathcal{K}} S_{m,k} = 0$, it indicates that UBS $m$ is underutilized and should be put to sleep, represented by $A_m = 0$; otherwise, $A_m = 1$, meaning UBS $m$ is active. Thus, the binary vector $\mathbf{A} = \left[ A_1, \cdots, A_M \right] $ represents the operating status of UBSs.
The key notations used in this paper are presented in Table~\ref{tab:notations}.

\subsection{Data Communication Model}
In this subsection, we present the channel model and derive the uplink data rate.
\subsubsection{Channel Modeling and Estimation}
The block fading model \cite{bjornsonMassiveMIMONetworks2017} is adopted, where the channel remains fixed within a finite-sized time-frequency coherence interval and is mutually independent across these intervals. Each coherence interval consists of $\tau_c$ symbols, with $\tau_p$ symbols dedicated to channel estimation, while the remaining $\tau_u = \tau_c - \tau_p$ symbols are allocated for uplink transmission.
We assume that the channel response $\mathbf{h}_{m,k}$ between UE $k$ and UBS $m$ follows the correlated Rayleigh fading model:
\begin{align}
	\mathbf{h}_{m,k} \sim \mathcal{CN}\left( \mathbf{0}_N, \mathbf{R}_{m,k} \right),
\end{align}
where the complex Gaussian distribution $\mathcal{CN}(\cdot, \cdot)$ models the small-scale fading. The matrix $\mathbf{R}_{m,k} \in \mathbb{C}^{N\times N}$ represents the spatial correlation between UE $k$ and UBS $m$, depicting both the large-scale fading and spatial channel correlation.
The BSs forward the received pilot signals to the edge cloud for channel estimation.
Employing the classic minimum mean square error method, we can derive the estimated channel $\hat{\mathbf{h}}_{m,k}$ of the practical channel $\mathbf{h}_{m,k}$, following the approach in \cite{sunJointUserAssociation2023}.

\subsubsection{Derivation of Uplink Rate}
In the uplink data transmission phase, UEs transmit data to UBSs. Each UBS receives a superposition of signals from all UEs, and the received signal at UBS $m$ is denoted as $\mathbf{y}_m \in \mathbb{C}^N$:
\begin{align}
	\mathbf{y}_m = \sum_{k \in \mathcal{K}} \mathbf{h}_{m,k}\sqrt{P_k}\varsigma_k+\mathbf{n}_m,
	\label{eq:received_signal}
\end{align}
where $\varsigma_k$ represents the transmitted signal of UE $k$, with $\mathbb{E}\left\{ \varsigma_k \right\} =0$ and $\mathbb{E}\left\{ \vert \varsigma_k \vert^2 \right\} =1$. $P_k\geq 0$ is the transmit power of UE $k$, and $\mathbf{n}_m \in \mathbb{C}^N \sim \mathcal{CN}\left( \mathbf{0}_N, \sigma^2 \mathbf{I}_N \right) $ denotes the additive Gaussian noise at UBS $m$.

We utilize the fully centralized operation of FD-RAN, where the received signals from all UBSs $\mathbf{y}_1, \cdots, \mathbf{y}_M$ are delivered to the edge cloud for further processing. Specifically, the edge cloud can design the combining vector $\mathbf{v}_{m,k} \in \mathbb{C}^{N},~ \forall m \in \mathcal{M}$ from a global perspective based on the global received signals. Consequently, the estimation of $\varsigma_{k}$ for each UBS, denoted as $\hat{\varsigma}_{m, k}$, can be obtained, and the global estimation $\hat{\varsigma}_{k}$ is derived as \cite{bjornsonMassiveMIMONetworks2017}:
\begin{align}
	\hat{\varsigma}_{k} = \sum_{m \in \mathcal{M}} \hat{\varsigma}_{m,k} = \sum_{m \in \mathcal{M}} S_{m,k} \mathbf{v}^H_{m,k} \mathbf{y}_{m}.
	\label{eq:estimation_of_signal}
\end{align}

Referring to \eqref{eq:received_signal} and \eqref{eq:estimation_of_signal}, utilizing the use-and-then-forget bound method \cite{bjornsonMassiveMIMONetworks2017}, we can derive the following uplink rate of UE $k$:
\begin{align}
	R_k\left( \mathbf{P}, \mathbf{S} \right)  = \frac{\tau_u}{\tau_u + \tau_p} B \log_2 \left( 1 + \mathrm{SINR}_k \right) ~ \mathrm{bit/s},
	\label{eq:uplink_rate}
\end{align}
where vector $\mathbf{P}$ denotes the transmit power of all UEs, $B$ indicates the allocated bandwidth, and $\mathrm{SINR}_k$ is expressed as:
\begin{align}
	\frac{P_k \left\vert \mathbb{E}\left\{  \mathrm{DS}_k \right\}  \right\vert^2}{\sum_{k^\prime \in \mathcal{K}} P_{k^\prime} \mathbb{E}\left\{ \left\vert \mathrm{IS}_{k,k^\prime} \right\vert^2 \right\} - P_k \left\vert \mathbb{E}\left\{ \mathrm{DS}_k \right\}  \right\vert^2 + \sigma^2\mathbb{E}\left\{ \mathrm{NS}_k \right\}},
	\label{SINR}
\end{align}
where $\mathrm{DS}_k, \mathrm{IS}_{k, k^\prime}, \mathrm{NS}_{m,k}$ are defined as the desired signal, the interference signal, and the noise signal, respectively, and can be denoted as:
\begin{align}
	\mathrm{DS}_k             & = \sum_{m \in \mathcal{M}} S_{m,k} \mathbf{v}_{m,k}^H \mathbf{h}_{m,k}, \label{eq:DS}        \\
	\mathrm{IS}_{k, k^\prime} & = \sum_{m \in \mathcal{M}} S_{m,k} \mathbf{v}_{m,k}^H \mathbf{h}_{m,k^\prime}, \label{eq:IS} \\
	\mathrm{NS}_{k}           & = \sum_{m \in \mathcal{M}} \Vert S_{m,k} \mathbf{v}_{m,k} \Vert^2. \label{eq:NS}
\end{align}
In this paper, we adopt the maximum-ratio combining to obtain closed-form expressions,
\begin{align}
	\mathbf{v}_{m,k} = \frac{\hat{\mathbf{h}}_{m,k}}{\sqrt{\mathbb{E} \{ \Vert \hat{\mathbf{h}}_{m,k} \Vert^2 \}}}, ~ \forall m \in \mathcal{M}, k \in \mathcal{K}.
\end{align}

\subsection{Power Consumption Model} \label{sec:uplink_FD-RAN_energy_model}
In this subsection, we develop a holistic and realistic power consumption model for uplink FD-RAN, encompassing four key components:
BSs, fronthauls, edge cloud, and UEs.
In this architecture, BSs, which include CBSs and UBSs, have analogous power consumption models. We simplify the power consumption of CBSs $P^{\text{CBS}}$ to a constant value. This is justified by the periodic and relatively stable nature of control signals, which contrasts with the significant fluctuations observed in data traffic. The power models for the UBSs and the remaining system components are detailed in the subsequent sections.

\subsubsection{Uplink Base Stations}
UBSs handle uplink data functions, and their power consumption comprises the radio frequency (RF) unit, baseband unit (BBU), and architectural costs.
The power consumption of the RF, $P_m^{\mathrm{RF}}$, is modeled as:
\begin{align}
	P_{m}^{\mathrm{RF}} = \sum_{j \in \mathcal{J}_{\mathrm{RF}}} P_{m}^{\mathrm{RF}_{j, \mathrm{ref}}} \prod_{x \in \mathcal{X}_{\mathrm{RF}}} \left( \frac{x_m^{\mathrm{act}}}{x_m^{\mathrm{ref}}} \right)^{s_m^{j, x}},
	\label{eq:RF}
\end{align}
where $P_{m}^{\mathrm{RF}_{j, \mathrm{ref}}}$ represents the reference power of the $j$-th sub-component. $\mathcal{J}_{\mathrm{RF}}$ denotes the set of sub-components of the RF, which includes frequency synthesis, clock generation, analog-to-digital converter, and others. $\mathcal{X}_{\mathrm{RF}} = \left\{ N, B, Q \right\}$ determines that the power consumption of the RF is a function of the number of antennas ($N$), bandwidth ($B$), and quantization ($Q$). Additionally, $x_m^{\mathrm{act}}$ and $x_m^{\mathrm{ref}}$ represent the actual and reference values of $x$, respectively. $s_m^{j, x}$ denotes the scaling exponent factor of $x$ for the $j$-th sub-component.
Similarly, the BBU power model, $P_m^{\mathrm{BBU}}$, depends on the same parameters as the RF model but also includes spectral efficiency ($Se$), load ($Ld$), and streams ($St$). Specifically, the load is defined as being proportional to the transmit rate of UBS $m$:
\begin{align}
	\frac{Ld_m}{Ld_m^{\mathrm{ref}}} \triangleq \frac{R_m^\prime}{R_{m, \mathrm{ref}}^{\prime}}, \label{eq:load_definition}
\end{align}
where $R_m^\prime$ represents the transmit rate of UBS $m$, and $R_{m, \mathrm{ref}}^{\prime}$ is the reference value. The relationship between the transmit rates of the UBSs and the UEs is shown as follows:
\begin{align}
	\sum_{m \in \mathcal{M}} A_m R_m^\prime = \sum_{k \in \mathcal{K}} R_{k}, \sum_{m \in \mathcal{M}} A_m R_{m,\mathrm{ref}}^\prime = \sum_{k \in \mathcal{K}} R_{k, \mathrm{ref}}. \label{eq:relation_R_k_R_m}
\end{align}
In a typical network scenario, most configurations are treated as fixed parameters, with the load being the primary variable. Therefore, in this paper, we assume that the parameters $N$, $B$, $Q$, $Se$, and $St$ of the UBSs are fixed and identically configured.
Additionally, by making reasonable assumptions, we can derive a more concise form for the power consumption of UBSs, as shown in the following lemma.

\begin{lemma} \label{lemma:UBS_power_rate_load}
	In FD-RAN, when all UBSs have identical configurations for the number of antennas, bandwidth, quantization, spectral efficiency, and streams, except for their load, the power consumption of UBSs can be rewritten as:
	\begin{align}
		\sum_{m \in \mathcal{M}} A_m P_m^{\mathrm{BS}} = \sum_{m \in \mathcal{M}} A_m P_m^{\mathrm{BS}_\mathrm{fix}} + P_\mathrm{trf} \sum_{k \in \mathcal{K}} \frac{R_k}{{R_{k, \mathrm{ref}}}}.
	\end{align}
	where the first term represents the fixed part of power consumption in all UBSs, and $P_{\mathrm{trf}}$ denotes a constant power coefficient.
\end{lemma}
\begin{proof}
	See Appendix \ref{app:proof_lemma_UBS_power_rate_load}.
\end{proof}

Consequently, we can derive a simplified expression for $P_m^{\mathrm{UBS}}$ based on \lemmref{lemma:UBS_power_rate_load}, which depends on variable load $R_k$.
The architecture costs primarily stem from the AC-DC main supply, DC-DC power supply, and cooling power losses \cite{auerHowMuchEnergy2011,dessetFlexiblePowerModeling2012}. These losses are estimated by scaling the power consumption of other components in UBSs and are modeled by the loss factors $\sigma_{\mathrm{MS}}$, $\sigma_{\mathrm{DC}}$, and $\sigma_{\mathrm{CO}}$, respectively. Therefore, the power consumption of UBS $m$ is summarized as follows:
\begin{align}
	P_{m}^{\mathrm{UBS}} = N_{m}^s \times \frac{P_m^{\mathrm{RF}}+P_m^{\mathrm{BBU}}}{\left( 1-\sigma_{\mathrm{MS}} \right)\left( 1-\sigma_{\mathrm{DC}} \right) \left( 1-\sigma_{\mathrm{CO}} \right) },
	\label{eq:power_UBS}
\end{align}
where $N_m^s$ is the number of sectors.
It is important to note that even in sleep mode, UBSs still consume power. The sleep power consumption of UBS $m$, denoted as $P_m^{\mathrm{sleep}}$, is modeled as a fraction of its idle power consumption \cite{debaillieFlexibleFutureProofPower2015}:
\begin{align}
	P_m^{\mathrm{sleep}} = \eta_s P_m^{\mathrm{UBS}} \left( Ld = 0 \right),
\end{align}
where $\eta_s$ is the scaling factor, and $Ld = 0$ represents zero load on UBS $m$, indicating it is idle.

\subsubsection{Fronthauls}
Taking inspiration from \cite{vanchienJointPowerAllocation2020}, we propose a power consumption model for the fronthaul between UBSs and the edge cloud, which comprises two components: load-independent and load-dependent. It can be expressed as follows:
\begin{align} \label{eq:power_fronthaul}
	\sum\limits_{m \in {\cal M} } {P_m^{{\rm{FH}}}} = \sum\limits_{m \in  {\cal M} } \left( {{A _m}} P_m^{{\rm{F}}{{\rm{H}}_{{\rm{fix}}}}} + {\Delta _m^{{\rm{F}}{{\rm{H}}_{{\rm{trf}}}}}} \sum\limits_{k \in {\cal K}} {R_k} \right),
\end{align}
where $P_m^{{\rm{FH}}_{{\rm{fix}}}}$ denotes the fixed power part, $\Delta_m^{{\rm{FH}}_{{\rm{trf}}}}$ represents the load-dependent power factor, and $R_k$ denotes the transmission rate of UE $k$.

\subsubsection{Edge Cloud}
We consider that a portion of BBU functions in UBSs is offloaded to the edge cloud for processing \cite{fioraniModelingEnergyPerformance2016}. Consequently, the power consumption of the edge cloud is expressed as follows
\footnote{In this paper, we do not consider BBU sleeping in the edge cloud, so the power consumption of the edge cloud is related to all BSs without sleeping.}:
\begin{align} \label{eq:power_EC}
	P_{\mathrm{EC}} = \kappa \theta \sum_{m \in \mathcal{M}}P_m^{\mathrm{UBS}}.
\end{align}
Accordingly, the power model of UBSs is updated based on \eqref{eq:power_UBS} as follows:
\begin{align} \label{power_all_BSs_new}
	\sum_{m \in \mathcal{M}} P_m^{\mathrm{UBS}} \leftarrow \Big( 1- \kappa \theta \Big) \sum_{m \in \mathcal{M} } P_m^{\mathrm{UBS}},
\end{align}
\begin{align}
	P_m^{\mathrm{sleep}} = \left( 1-\kappa \theta \right) P_m^{\mathrm{sleep}},
\end{align}
where $0 < \kappa \leq 1$ represents the level of centralization in FD-RAN, indicating the proportion of transferred functions from the BBU to the edge cloud. Meanwhile, $\theta$ is the power percentage of BBUs in the UBSs, which can be calculated as:
\begin{align}
	\theta = \left. \frac{\psi_d \sum_{m \in \mathcal{M} } N_{m}^s P_m^{\mathrm{BBU}}}{\left( 1-\sigma_{\mathrm{MS}} \right) \left( 1-\sigma_{\mathrm{DC}} \right) \left( 1-\sigma_{\mathrm{CO}} \right)} \middle/ \sum_{m \in \mathcal{M} } P_m^{\mathrm{UBS}} \right. .
\end{align}
In addition, centralized operations in the edge cloud offer more energy-efficient approaches, benefiting from stacking gain, pooling gain, and cooling gain \cite{fioraniModelingEnergyPerformance2016}.
Stacking gain occurs because centralized BBUs can utilize processing resources more efficiently than distributed BBUs, denoted by $\zeta > 1$. Moreover, centralized operations enable the incorporation of more energy-efficient BBUs within the edge cloud, leading to pooling gain, which is modeled as a $\lambda$-fold increase in BBU capacity but at the cost of $\xi$ times higher power consumption.
Given these two factors, the power model of edge cloud is updated based on \eqref{eq:power_EC} as follows:
\begin{align} \label{eq:power_EC_new}
	P_{\mathrm{EC}} \leftarrow P_{\mathrm{EC}} \times \frac{\xi}{\vert \mathcal{M} \vert} \left\lceil \frac{\vert \mathcal{M} \vert}{\lambda \zeta} \right\rceil.
\end{align}
Additionally, more efficient cooling systems are employed in centralized operations, reflected by the factor $\varrho$. Consequently, the power model of edge cloud is further updated based on \eqref{eq:power_EC_new} as:
\begin{align}
	P_{\mathrm{EC}}\leftarrow \begin{cases}
		                          P_{\mathrm{EC}} \left( \left. \sigma_{\mathrm{CO}}^{\prime} \middle/ \varrho \right. + 1-\sigma_{\mathrm{CO}}^{\prime} \right),                               & \sigma_{\mathrm{CO}}\neq 0. \\
		                          P_{\mathrm{EC}} \left( \left.\sigma_{\mathrm{CO}}^{\prime} \middle/ \left( \left( 1-\sigma_\mathrm{{CO}}^{\prime} \right) \varrho \right) +1 \right. \right), & \sigma_{\mathrm{CO}} = 0.
	                          \end{cases}
\end{align}
where $\sigma_\mathrm{{CO}}^{\prime}$ represents the cooling loss in the edge cloud. It is worth noting that when $\sigma_{\mathrm{CO}} = 0$, indicating the absence of an active cooling system in the BSs, the stacked BBUs in the edge cloud still require cooling. As a result, the cooling gain becomes negative in this scenario.

\subsubsection{User Equipments}
We propose the following model for UE $k$, including circuit power and transmit power \cite{limEnergyefficientBestselectRelaying2012}:
\begin{align}
	P_{k}^{\mathrm{UE}} = P_{k}^{\mathrm{UE}_{\mathrm{cp}}} + \Delta_k^{\mathrm{UE_{pa}}} P_k,
\end{align}
where $P_k$ represents the transmit power, $\Delta_k^{\mathrm{UE_{pa}}}\geq 1$ denotes the power efficiency of the PA in UE, and $P_{k}^{\mathrm{UE}_{\mathrm{cp}}}$ corresponds to the circuit power.

\subsubsection{Holistic Power Consumption}
By aggregating all the aforementioned power consumption components, the holistic power consumption of uplink FD-RAN can be summarized as:
\begin{align}
	 & P_{\mathrm{N}}  \left( \mathbf{P}, \mathbf{A}, \mathbf{S} \right) = \sum_{m \in \mathcal{M}} \left( A_m P_m^{\mathrm{UBS}}\left( \left\{ R_k \right\} \right) + \left( 1-A_m \right) P_m^{\mathrm{sleep}} \right)\notag                                \\
	 & + \sum_{m \in \mathcal{M}} P_m^{\mathrm{FH}} \left( \mathbf{A}, \left\{ R_k \right\}  \right) + P_{\mathrm{EC}} \left(\left\{ R_k \right\}  \right) + \sum_{k \in \mathcal{K}} P_{k}^{\mathrm{UE}} \left( \mathbf{P} \right), \label{eq:uplink_energy}
\end{align}
where the parentheses indicate that power consumption is a function of the corresponding variables. $P_{\mathrm{N}}$ is a function of $\mathbf{P}$, $\mathbf{A}$, and $\mathbf{S}$, with $\left\{ R_k \right\}$ being a function of $\mathbf{P}$ and $\mathbf{S}$, where $\left\{ R_k \right\}$ denotes the set of UEs.

\begin{lemma} \label{lemm:affine_P_N}
	$P_N$ is an affine function of $\left\{ R_k \right\}$ with positive coefficients.
\end{lemma}
\begin{proof}
	This proof can readily be inferred from equations \eqref{eq:power_fronthaul}, \eqref{eq:power_EC}-\eqref{power_all_BSs_new}, and \eqref{eq:uplink_energy}.
\end{proof}

\section{Problem Formulation} \label{sec:problem_formulation}
\subsection{Energy Efficiency Maximization Problem}
Based on the system models in Section~\ref{sec:system_model}, we formulate the problem of maximizing energy efficiency in uplink FD-RAN subject to QoS constraints as follows:
\begin{subequations} \label{P}
	\begin{align}
		\mathcal{P}\text{:} ~ & \max_{\mathbf{P}, \mathbf{S}, \mathbf{A}} ~ \mathrm{EE}\left( \mathbf{P}, \mathbf{S}, \mathbf{A} \right) = \frac{\sum_{k \in \mathcal{K}} R_k\left( \mathbf{P}, \mathbf{S} \right)}{P_{\mathrm{N}}\left(\mathbf{P}, \mathbf{S}, \mathbf{A}\right)} \label{P:obj} \\
		\rm{s.t.}             & ~ R_k\left(\mathbf{P}, \mathbf{S} \right)  \geq R_{k,\min}, ~ \forall k \in \mathcal{K}, \label{P:st_rate}                                                                                                                                                       \\
		                      & ~ 0 \leq P_k \leq P_{\max}, ~ \forall k \in \mathcal{K}, \label{P:st_p}                                                                                                                                                                                          \\
		                      & ~ \left\vert \mathcal{M}_k \right\vert \leq L, ~ \forall k \in \mathcal{K}, \label{P:st_UE_connect}                                                                                                                                                              \\
		                      & ~ \left\vert \mathcal{K}_m \right\vert \leq N, ~ \forall m \in \mathcal{M}, \label{P:st_BS_connect}                                                                                                                                                              \\
		                      & ~ A_m = \max \left\{ S_{m,k} \right\}, ~ \forall m \in \mathcal{M}, k \in \mathcal{K},\label{P:st_connect_sleep}                                                                                                                                                 \\
		                      & ~ S_{m,k}, A_m \in \left\{ 0,1 \right\}, ~ \forall m \in \mathcal{M}, k \in \mathcal{K}, \label{P:st_binary}
	\end{align}
\end{subequations}
where \eqref{P:st_rate} denotes the minimum rate constraint on UE $k$ to ensure QoS while optimizing energy efficiency. Constraint \eqref{P:st_p} represents the power constraint of UE $k$. Constraints \eqref{P:st_UE_connect} and \eqref{P:st_BS_connect} specify the maximum number of associations for UEs and UBSs, respectively. Specifically, each UE can communicate with at most $L$ UBSs, and each UBS can serve at most $N$ UEs. These constraints are reasonable considering the computational capabilities of UBSs and UEs, fronthaul capacity limitations, and the scalability of FD-RAN.
Constraint \eqref{P:st_binary} indicates that variables $S_{m,k}$ and $A_m$ are binary variables, and \eqref{P:st_connect_sleep} represents the relationship between $S_{m,k}$ and $A_m$. Particularly, if a UBS is in sleep mode, it cannot serve any UEs.

\subsection{Iterative Optimization}
The formulated problem $\mathcal{P}$ involves the binary UEs-UBSs association matrix $\mathbf{S}$, binary UBS operating status vector $\mathbf{A}$, and continuous UE power vector $\mathbf{P}$, resulting in a non-convex nature for both the objective function and constraints. This problem $\mathcal{P}$ is an MINLP, which is known to be NP-hard.

$\mathcal{P}$ can be decomposed into two subproblems.
\begin{subequations}
	\label{P_u}
	\begin{align}
		\mathcal{P}_u \text{:} ~ & \max_{\mathbf{S}, \mathbf{A}} ~ \mathrm{EE}\left( \mathbf{P}^\circ, \mathbf{S}, \mathbf{A} \right) = \frac{\sum_{k \in \mathcal{K}} R_k\left( \mathbf{P}^\circ, \mathbf{S} \right)}{P_{\mathrm{N}} \left(\mathbf{P}^\circ, \mathbf{S}, \mathbf{A}\right)} \label{P_u:obj} \\
		\rm{s.t.}                & ~ R_k\left(\mathbf{P}^\circ, \mathbf{S} \right)  \geq R_{k,\min}, ~ \forall k \in \mathcal{K},                                                                                                                                                                            \\
		                         & ~ \eqref{P:st_UE_connect}-\eqref{P:st_binary}, \notag \label{P_u:sts}
	\end{align}
\end{subequations}
The first subproblem involves joint optimization of UE association and UBS sleeping, where $\mathbf{P}^\circ$ is the UE power vector obtained by solving the second subproblem.
\begin{subequations}
	\label{P_l}
	\begin{align}
		\mathcal{P}_l \text{:} ~ & \max_{\mathbf{P}} ~ \mathrm{EE}\left( \mathbf{P}, \mathbf{S}^\circ, \mathbf{A}^\circ \right) = \frac{\sum_{k \in \mathcal{K}} R_k\left( \mathbf{P}, \mathbf{S}^\circ \right)}{P_{\mathrm{N}} \left(\mathbf{P}, \mathbf{S}^\circ, \mathbf{A}^\circ \right)} \label{P_l:obj} \\
		\rm{s.t.}                & ~ R_k\left(\mathbf{P}, \mathbf{S}^\circ \right)  \geq R_{k,\min}, ~ \forall k \in \mathcal{K},                    \label{P_l:st_rate}                                                                                                                                      \\
		                         & ~ \eqref{P:st_p}, \notag
	\end{align}
\end{subequations}
which addresses the optimization of power control, given an optimized $\mathbf{S}^\circ$ and $\mathbf{A}^\circ$ obtained from the first subproblem.
We will address the first subproblem $\mathcal{P}_l$ in Section~\ref{sec:SLMDB} and the second subproblem $\mathcal{P}_u$ in  Section~\ref{sec:joint_alg}, respectively.

\section{SLMDB Algorithm for Power Control} \label{sec:SLMDB}
In this section, our focus is on resolving the non-convex subproblem $\mathcal{P}_l$ by proposing the SLMDB algorithm.

As presented in \eqref{P_l}, $\mathcal{P}_l$ represents a continuous yet non-convex problem. Notably, the functions $R_k\left(\mathbf{P} \right)$\footnote{For the sake of notational simplicity, we omit $\mathbf{S}^\circ$ and $\mathbf{A}^\circ$ from the expressions in this section.} in constraints \eqref{P_l:st_rate} are non-convex with respect to the variable $\mathbf{P}$.
In fact, constraints \eqref{P_l:st_rate} can be equivalently expressed as the following convex constraints:
\begin{subequations}
	\begin{align}
		r_k \left( \mathbf{P} \right) \leq 0, ~ \forall k \in \mathcal{K},
	\end{align}
	\begin{align}
		r_k \left( \mathbf{P} \right) = ~ & \gamma_k \left( \sum_{k^\prime \in \mathcal{K}} P_{k^\prime} \mathbb{E}\left\{ \left\vert \mathrm{IS}_{k, k^\prime} \right\vert^2 \right\} + \sigma^2 \mathbb{E} \left\{ \mathrm{NS}_k \right\}  \right) \notag \\
		                                  & - \left( 1+\gamma_k \right) P_k \left\vert \mathbb{E} \left\{ \mathrm{DS}_k \right\}  \right\vert^2,
	\end{align}
	\begin{align}
		\gamma_k = 2^{\tau_c R_{k,\min} / \left( \tau_u B \right) } - 1,
	\end{align} \label{eq:st_convex_rate}
\end{subequations}
then the problem $\mathcal{P}_l$ can be reformulated as follows:
\begin{align}
	\mathcal{P}_l\text{:} ~ & \max_{\mathbf{P}} ~ \mathrm{EE}\left( \mathbf{P} \right) = \frac{\sum_{k \in \mathcal{K}} R_k\left( \mathbf{P} \right)}{P_{\mathrm{N}} \left( \mathbf{P}, \left\{ R_k\left( \mathbf{P} \right)  \right\}\right)} \label{P_l_new:obj} \\
	\rm{s.t.}               & ~ \eqref{P:st_p}, \notag \eqref{eq:st_convex_rate},
\end{align}
where $\left\{ R_k \left( \mathbf{P} \right)  \right\} $ inside the objective function \eqref{P_l_new:obj} is to declare that it is an explicit function of $\left\{ R_k \left( \mathbf{P} \right)  \right\} $.
However, it is still non-convex considering the non-convexity of objective function \eqref{P_l_new:obj}, thus solving directly is intractable.

\begin{lemma} \label{lemm:parametric_problem_optimality_of_FP}
	The maximum energy efficiency of the network, denoted as the optimal value $\pi^{\ast}$ of $\mathcal{P}_l$, is achieved if and only if
	\begin{align}
		F\left( \pi^\ast \right) & = \max_{\mathbf{P}} ~ \sum_{k \in \mathcal{K}} R_k\left( \mathbf{P} \right) - \pi^{\ast} P_{\mathrm{N}}^u \left( \mathbf{P}, \left\{ \hat{R}_k\left( \mathbf{P} \right)  \right\}\right) \notag \\
		                         & = \sum_{k \in \mathcal{K}} R_k\left( \mathbf{P}^{\ast} \right) - \pi^{\ast} P_{\mathrm{N}}^u \left( \mathbf{P}^{\ast}, \left\{ \hat{R}_k\left( \mathbf{P}^{\ast} \right)  \right\}\right) = 0,
	\end{align}
	where $\pi^{\ast} = \mathrm{EE} \left( \mathbf{P}^{\ast} \right) = \max_{\mathbf{P}} \mathrm{EE} \left( \mathbf{P}  \right)$.
\end{lemma}
\begin{proof}
	See Appendix \ref{app:proof_lemm_parametric_problem_optimality_of_FP}.
\end{proof}

In reality, problem $\mathcal{P}_l$ falls under the category of a nonlinear fractional programming problem, as indicated in \lemmref{lemm:parametric_problem_optimality_of_FP}, which is equivalent to a parametric problem. Nevertheless, the quest to solve the equivalent problem remains arduous due to its non-convex nature.
A more manageable version of fractional programming is the concave-convex fractional programming (CCFP) \cite{shenFractionalProgrammingCommunication2018}, characterized by a concave numerator and a convex denominator in the context of minimization problems. To this end, we iteratively approximate $\mathcal{P}_l$ by solving a conservative lower-bound CCFP problem at each step, thereby ensure worst-case energy efficiency.

\subsection{Successive Lower-Bound Maximization Algorithm}
The numerator and denominator in \eqref{P_l_new:obj} exhibit non-convexity due to the non-convex nature of $R_k\left( \mathbf{P} \right)$. It is worth noting that $R_k\left( \mathbf{P} \right)$ can be represented as the difference of two concave functions, as follows:
\begin{small}
	\begin{subequations}
		\renewcommand{\theequation}{\normalsize\theparentequation\alph{equation}}
		\begin{align} \label{eq:rate}
			R_k\left( \mathbf{P} \right) & = \frac{\tau_u}{\tau_u+\tau_p} B \left( f_k(\mathbf{P})- g_k(\mathbf{P}) \right),
		\end{align}
		\begin{align}
			f_k (\mathbf{P}) & = \log_2\left( \sum_{k^\prime \in \mathcal{K}} P_{k^\prime} \mathbb{E}\left\{ \left\vert \mathrm{IS}_{k, k^\prime}\right\vert^2 \right\} + \sigma^2\mathbb{E}\left\{ \mathrm{NS}_k \right\} \right),                      \\
			g_k (\mathbf{P}) & = \log_2\Bigg( \sum_{k^\prime \in \mathcal{K}} P_{k^\prime} \mathbb{E}\left\{ \left\vert \mathrm{IS}_{k,k^\prime} \right\vert^2 \right\} - P_k \left\vert \mathbb{E}\left\{  \mathrm{DS}_k \right\}  \right\vert^2 \notag \\
			                 & \hspace{4em} + \sigma^2\mathbb{E}\left\{ \mathrm{NS}_k \right\} \Bigg).
		\end{align}
	\end{subequations}
\end{small}

To transform problem $\mathcal{P}_l$ into a more tractable CCFP form and ensure worst-case energy efficiency, we first derive a convex upper bound of $R_k\left( \mathbf{P} \right)$:
\begin{align} \label{eq:hat_rate}
	R_k\left( \mathbf{P} \right) & \leq \hat{R}_k\left( \mathbf{P}, \mathbf{P}^{(n)} \right) \notag                                          \\
	                             & = \frac{\tau_u}{\tau_u+\tau_p} B \left( \hat{f}_k(\mathbf{P}, \mathbf{P}^{(n)})-g_k (\mathbf{P}) \right),
\end{align}
where $\hat{f}_k(\mathbf{P}, \mathbf{P}^{(n)}) \geq f_k(\mathbf{P})$ defined in \eqref{eq:f_SCA}, is the first-order Taylor expansion of $f_k(\mathbf{P})$ at the point $\mathbf{P}^{(n)}$, and $\mathbf{P}^{(n)}$ is the fixed point of the Taylor expansion in the $n$-th iteration.

Similarly, we derive a concave lower bound of $R_k\left( \mathbf{P} \right)$ as:
\begin{align} \label{eq:overline_rate}
	R_k\left( \mathbf{P} \right) & \geq \overline{R}_k\left( \mathbf{P}, \mathbf{P}^{(n)} \right)                                              \\
	                             & = \frac{\tau_u}{\tau_u+\tau_p} B \left( f_k (\mathbf{P}) - \hat{g}_k(\mathbf{P}, \mathbf{P}^{(n)}) \right),
\end{align}
where $\hat{g}_k(\mathbf{P}, \mathbf{P}^{(n)}) \leq g_k(\mathbf{P})$ defined in \eqref{eq:g_SCA}, is the first-order Taylor expansion of $g_k(\mathbf{P})$ at $\mathbf{P}^{(n)}$.

\begin{figure*}[!t]
	\begin{align}
		\hat{f}_k\big(\mathbf{P}, \mathbf{P}^{(n)}\big) & = \sum_{k^{\prime} \in \mathcal{K}} \frac{\mathbb{E}\left\{ \left\vert \mathrm{IS}_{k, k^\prime} \right\vert^2 \right\}}{\ln 2 \left( \sum_{k^\prime \in \mathcal{K}} P_{k^\prime}^{(n)} \mathbb{E}\left\{ \left\vert \mathrm{IS}_{k,k^\prime} \right\vert^2 \right\} + \sigma^2\mathbb{E}\left\{ \mathrm{NS}_k \right\} \right)} \times \left( P_{k^\prime} - P_{k^\prime}^{(n)} \right) + f_k  \left( \mathbf{P}^{(n)} \right), \label{eq:f_SCA}
	\end{align}
	\begin{align}
		\hat{g}_k\big(\mathbf{P}, \mathbf{P}^{(n)}\big) = \frac{\sum_{k^{\prime} \in \mathcal{K}} \mathbb{E}\left\{ \left\vert \mathrm{IS}_{k,k^\prime} \right\vert^2 \right\}\left( P_{k^\prime} - P_{k^\prime}^{(n)} \right) - \left\vert \mathbb{E}\left\{ \mathrm{DS}_k \right\}  \right\vert^2 \left( P_{k} - P_{k}^{(n)} \right)}{\ln 2 \left( \sum_{k^\prime \in \mathcal{K}} P_{k^\prime}^{(n)} \mathbb{E}\left\{ \left\vert \mathrm{IS}_{k,k^\prime} \right\vert^2 \right\} - P_k^{(n)}\left\vert \mathbb{E}\left\{ \mathrm{DS}_k \right\} \right\vert^2 + \sigma^2\mathbb{E}\left\{ \mathrm{NS}_k \right\} \right)} + g_k  \left( \mathbf{P}^{(n)} \right). \label{eq:g_SCA}
	\end{align}
	\hrulefill
\end{figure*}

By replacing the instances of $R_k\left( \mathbf{P} \right)$ in both the numerator and denominator of \eqref{P_l_new:obj} with the upper bound $\overline{R}_k\left( \mathbf{P}, \mathbf{P}^{(n)} \right)$ and lower bound $\hat{R}_k\left( \mathbf{P}, \mathbf{P}^{(n)} \right)$, respectively, we obtain a lower bound for the original problem $\mathcal{P}_l$ in the form of the following CCFP problem:
\begin{align}
	\mathcal{P}_l^{\prime}\text{:} ~ & \max_{\mathbf{P}} ~ \overline{\mathrm{EE}}\left( \mathbf{P}, \mathbf{P}^{(n)} \right) = \frac{\sum_{k \in \mathcal{K}} \overline{R}_k\left( \mathbf{P}, \mathbf{P}^{(n)} \right)}{P_{\mathrm{N}} \left( \mathbf{P}, \mathbf{P}^{(n)}, \left\{ \hat{R}_k\left( \mathbf{P}, \mathbf{P}^{(n)} \right)  \right\}\right)} \label{SP_1:obj} \\
	\rm{s.t.}                        & ~ \eqref{P:st_p}, \notag \eqref{eq:st_convex_rate}.
\end{align}
where $\overline{\mathrm{EE}}\left( \mathbf{P}, \mathbf{P}^{(n)} \right) \leq \mathrm{EE}\left( \mathbf{P} \right)$.

Instead of solving the intractable non-convex problem $\mathcal{P}_l$, we solve a sequence of approximated problems $\mathcal{P}_l^{\prime}$. By iteratively solving the problem $\mathcal{P}_l^{\prime}$, we can gradually improve the conservative approximation based on the optimal solution in the previous iteration. Specifically, we update $\mathbf{P}^{(n)}$ by solving the following problem:
\begin{align}
	\mathbf{P}^{(n)} = \arg \max_{\mathbf{P}} \overline{\mathrm{EE}} \left( \mathbf{P}, \mathbf{P}^{(n-1)} \right),\label{eq:power_update}
\end{align}
where $\mathbf{P}^{(n-1)} = \mathbf{P}^{(n-1), \ast}$ is the optimal power solution obtained in the $(n-1)$-th iteration.
As shown in \lemmref{lemm:EE_lower_bound} and \theoref{theo:SLM_stationary_convergent}, the proposed algorithm is both effective and globally convergent.
\begin{lemma} \label{lemm:EE_lower_bound}
	The function $\overline{\mathrm{EE}}\left( \mathbf{P}, \mathbf{P}_0 \right)$ serves as a global lower-bound for $\mathrm{EE}\left( \mathbf{P} \right)$, and equality holds if and only if $\mathbf{P} = \mathbf{P}_0$.
\end{lemma}
\begin{proof}
	See Appendix \ref{app:proof_lemm_EE_lower_bound}.
\end{proof}
\begin{theorem} \label{theo:SLM_stationary_convergent}
	Every limit point of the iterates generated by the Successive Lower-Bound Maximization (SLM) algorithm is a stationary point of $\mathcal{P}_l^{\prime}$, and the SLM algorithm is globally convergent.
\end{theorem}
\begin{proof}
	See Appendix \ref{app:proof_theo_SLM_stationary_convergent}.
\end{proof}

\subsection{Dinkelbach's Algorithm} \label{sec:Dinkelbach}
\begin{lemma} \label{lemma:CCFP}
	Problem $\mathcal{P}_l^{\prime}$ satisfies the standard CCFP formulation.
\end{lemma}
\begin{proof}
	Since $\overline{R}_k\left( \mathbf{P}, \mathbf{P}^{(n)} \right)$ is concave, the sum of concave functions in the numerator is also concave. Furthermore, based on \lemmref{lemm:affine_P_N} and the convexity of $\hat{R}_k\left( \mathbf{P}, \mathbf{P}^{(n)} \right)$, the denominator is convex. Additionally, it is evident that the denominator $P_N \geq 0$, and the constraints define a feasible convex domain. Therefore, according to the definition in \cite{shenFractionalProgrammingCommunication2018}, $\mathcal{P}_l^{\prime}$ is a CCFP problem.
\end{proof}

According to \lemmref{lemma:CCFP}, $\mathcal{P}_l^\prime$ is a standard CCFP problem. Its objective function is pseudoconcave, implying that any stationary point is a global maximum point. Consequently, it can be solved using various algorithms \cite{ishedenFrameworkLinkLevelEnergy2012}. To efficiently address the CCFP problem, we employ Dinkelbach's algorithm \cite{dinkelbachNonlinearFractionalProgramming1967} to obtain the globally optimal solution of $\mathcal{P}_l^{\prime}$. Specifically, we solve the following equivalent problem:
\begin{align}
	\mathcal{P}_l^{\prime\prime}: ~ & \max_{\mathbf{P}} ~ U\left( \mathbf{P}, \mathbf{P}^{(n)} \right) = \sum_{k \in \mathcal{K}} \overline{R}_k\left( \mathbf{P}, \mathbf{P}^{(n)} \right) \notag           \\
	                                & \hspace{2.3em} - \pi P_{\mathrm{N}} \left( \mathbf{P}, \mathbf{P}^{(n)} \left\{ \hat{R}_k\left( \mathbf{P}, \mathbf{P}^{(n)} \right)  \right\}\right)  \label{P_2:obj} \\
	\rm{s.t.}                       & ~ \eqref{P:st_p}, \notag\eqref{eq:st_convex_rate},
\end{align}
where the problem is a parametric subtractive problem that is strictly convex in $\mathbf{P}$. The parameter $\pi \geq 0$ is the fractional parameter updated at each iteration of the Dinkelbach's algorithm. In the $\tilde{n}$-th iteration, $\pi$ is updated as follows:
\begin{align} \label{eq:pi}
	\pi = \frac{\sum_{k \in \mathcal{K}} \overline{R}_k\left( \mathbf{P}^{(\tilde{n})}, \mathbf{P}^{(n)} \right)}{P_{\mathrm{N}} \left( \mathbf{P}^{(\tilde{n})}, \mathbf{P}^{(n)} \left\{ \hat{R}_k\left( \mathbf{P}^{(\tilde{n})}, \mathbf{P}^{(n)} \right)  \right\}\right)}.
\end{align}
By iteratively solving the equivalent problem $\mathcal{P}_l^{\prime\prime}$ of the CCFP problem $\mathcal{P}_l^{\prime}$, we can obtain the globally optimal solution for $\mathcal{P}_l^{\prime}$, as demonstrated in \theoref{theo:Dinkelbach}.
The solution for the original problem $\mathcal{P}_l$ is now complete, and the comprehensive algorithm, are outlined in \algref{alg:SLMDB}.
\begin{theorem} \label{theo:Dinkelbach}
	Dinkelbach's Algorithm converges to the globally optimal solution of $\mathcal{P}_l^{\prime}$.
\end{theorem}
\begin{proof}
	\lemmref{lemm:parametric_problem_optimality_of_FP} reveals that the global optimal solution of the CCFP problem $\mathcal{P}_l^{\prime}$ can be obtained by finding the root of the nonlinear function $F(\pi)$. Since Dinkelbach's Algorithm utilizes a root-finding method, the optimality of $\mathcal{P}_l^{\prime\prime}$ is guaranteed. Furthermore, the convergence of Dinkelbach's Algorithm to the optimal solution has been proven in \cite{dinkelbachNonlinearFractionalProgramming1967}. Hence, the global optimality of Dinkelbach's Algorithm is established.
\end{proof}

\begin{algorithm}[!t]
	\caption{SLMDB algorithm for $\mathcal{P}_l$}
	\label{alg:SLMDB}
	\KwIn{Problem $\mathcal{P}_l$, prescribed threshold $\epsilon=10^{-3}$, feasible initial power $\mathbf{P}^{(0)}$, initial energy efficiency $\mathrm{EE}^{(n)} = \inf$, initial difference of objective function $\vartheta^{(0)} = \epsilon + 1$.}
	\KwOut{Optimal power $\mathbf{P}^{\ast}$ for the problem $\mathcal{P}_l$.}
	Set $n = 0$\;
	\While{$\vartheta^{(n)} > \epsilon$}{
		Set $n = n + 1$\;
		Update $\mathbf{P}^{(n)}$ according to \eqref{eq:power_update} using Dinkelbach's Algorithm\;
		Calculate the energy efficiency $\mathrm{EE}^{(n)}$ with power $\mathbf{P}^{(n)}$\;
		Compute the difference $\vartheta^{(n)} = \left.\left( \mathrm{EE}^{(n)} - \mathrm{EE}^{(n-1)} \right)  \middle/ \mathrm{EE}^{(n-1)}\right.$\;
	}
	Assign $\mathbf{P}^{\ast} = \mathbf{P}^{(n)}$.
\end{algorithm}

\section{TriMSM Algorithm for UE Association and UBS Sleeping} \label{sec:joint_alg}

In this section, we tackle the nonlinear integer programming subproblem $\mathcal{P}_u$ using the TriMSM algorithm alongside three low-complexity realizations.

\subsection{Modified Many-to-Many Swap Matching}
The joint UE association and UBS sleeping can be simplified into a sole UE association problem, where the UBSs' operational status is defined by the UEs' associations based on the constraints outlined in \eqref{P:st_connect_sleep}.
Considering the intricate associations between UEs and UBSs, this can be formulated as a matching game in which UEs and UBSs belong to two separate sets, denoted as $\mathcal{K}$ and $\mathcal{M}$, respectively. These players act rationally to make decisions that maximize their individual interests. In FD-RAN, players have the capability to exchange information among themselves via the CBS, granting them complete information about the game. The formal definition of the many-to-many matching is presented as follows:
\begin{definition}[Many-to-Many Matching] \label{def:many_to_many_matching}
	For two disjoint sets $\mathcal{M}$ and $\mathcal{K}$, a many-to-many matching, denoted as $\mathcal{S} \subseteq \mathcal{M}\times \mathcal{K}$, is a mapping from the set $\mathcal{M} \cup \mathcal{K}$ into the set of all subsets of $\mathcal{M} \cup \mathcal{K}$, such that for each $K_i \in \mathcal{K}$ and $M_j \in \mathcal{M}$, the following conditions hold:
	\begin{enumerate}
		\item $\mathcal{S}(K_i) \subseteq \mathcal{M}$, and in particular, $\mathcal{S}(K_i) = \emptyset$ if $K_i$ is not matched to any $M_j$;
		\item $\mathcal{S}(M_j) \subseteq \mathcal{K}$, and in particular, $\mathcal{S}(M_j) = \emptyset$ if $M_j$ is not matched to any $K_i$;
		\item $\left\vert \mathcal{S}(K_i) \right\vert \leq T$;
		\item $\left\vert \mathcal{S}(M_j) \right\vert \leq N$;
		\item $K_i \in \mathcal{S}(M_j)$ if and only if $M_j \in \mathcal{S}(K_i)$.
	\end{enumerate}
\end{definition}
In this definition, Condition (1) specifies UBSs matched with UE $K_i$ as a subset of $\mathcal{M}$, while Condition (2) indicates UEs matched with UBS $M_j$ as a subset of $\mathcal{K}$. Conditions (3) and (4) set the maximum matching pairs for players $K_i$ and $M_j$, aligning with constraints \eqref{P:st_UE_connect} and \eqref{P:st_BS_connect}. Condition (5) denotes the inherent reciprocity in matching pairs.

The matching game formulated in this paper is a many-to-many matching with externalities \cite{bodine-baronPeerEffectsStability2011,yuFullydecoupledRadioAccess2023}, where peer effects arise due to interference caused by co-channel transmission.
Since the matching results significantly depend on competition among players, we define the following preference lists of players as criteria for decision-making in the matching game:
\begin{definition}[Modified Preference Lists]
	For UE $K_i$, there exist two distinct UBSs $M_j$ and $M_{j^\prime}$\footnote{In this context, $M_j^{\prime}$ can be an empty set ($\emptyset$). Consequently, players matched with UE $K_i$ can be added or removed, allowing for more flexible matchings. It's important to emphasize that this addition operation does not violate the \defref{def:many_to_many_matching}, as it cannot result in the formation of a matching $\mathcal{S}^{\prime}$.}, each forming separate matchings denoted as $\mathcal{S}$ and $\mathcal{S}^\prime$, where $M_j \in \mathcal{S}(K_i)$ and $M_{j^\prime} \in \mathcal{S}^\prime(K_i)$. We denote the preference notation $\succ_{K_i}$ for UE $K_i$, and define its preference for UBSs as follows:
	\begin{align}
		\left( M_j, \mathcal{S} \right) \succ_{K_i} \left( M_{j^\prime}, \mathcal{S}^\prime \right) \Leftrightarrow \begin{cases}
			                                                                                                            \mathrm{EE}\left( \mathcal{S} \right) > \mathrm{EE}\left( \mathcal{S}^\prime \right), \\
			                                                                                                            R_k\left( \mathcal{S} \right) \geq R_{k, \min},~\forall k \in \mathcal{K},            \\
		                                                                                                            \end{cases}
	\end{align}
	which implies that UE $K_i$ prefers $M_j$ over $M_{j^\prime}$ only if $\mathcal{S}$ would yield higher energy efficiency than $\mathcal{S}^\prime$, and all UEs can attain the minimum QoS rate when $\mathcal{S}$. It is crucial to emphasize that, unless $\mathcal{S}^{\prime}$ exhibits superior energy efficiency and meets the QoS requirements, $\mathcal{S}^{\prime}$ is not the preferred choice.
	Similarly, for UBS $M_j$, with two different UEs and their corresponding formed matchings, $K_i \in \mathcal{S}(M_j)$ and $K_{i^\prime} \in \mathcal{S}^\prime(M_j)$, its preference for UEs is defined as:
	\begin{align}
		\left( K_i, \mathcal{S} \right) \succ_{M_j} \left( K_{i^\prime}, \mathcal{S}^\prime \right) \Leftrightarrow \begin{cases}
			                                                                                                            \mathrm{EE}\left( \mathcal{S} \right) > \mathrm{EE}\left( \mathcal{S}^\prime \right), \\
			                                                                                                            R_k\left( \mathcal{S} \right) \geq R_{k, \min},~\forall k \in \mathcal{K},            \\
		                                                                                                            \end{cases}
	\end{align}
	which indicates that UBS $M_j$ prefers $K_i$ if and only if $\mathcal{S}$ leads to higher energy efficiency with guaranteed QoS for all UEs.
\end{definition}
Different from the preference lists in traditional matching \cite{diSubChannelAssignmentPower2016}, the modified preference lists presented in this paper also accommodate the QoS constraints, as defined in \eqref{P:st_rate}, instead of merely comparing the objective function.

However, compared to classic two-sided matching, addressing many-to-many matching with externalities poses significant challenges and intricacies, rendering traditional approaches inapplicable directly \cite{diSubChannelAssignmentPower2016}.
In light of that, we pivot towards swap matching as a means to attain two-sided exchange stability and optimize the energy efficiency of the FD-RAN. The specific definition is articulated below:

\begin{definition}[Swap Matching]
	Given a matching $\mathcal{S}$ with $K_i \in \mathcal{S}(M_m), K_j \in \mathcal{S}(M_n), K_i \notin \mathcal{S}(M_n)$ and $K_j \notin \mathcal{S}(M_m)$, the swap matching, denoted as $\mathcal{S}_{jn}^{im}$, is defined as $\mathcal{S}_{jn}^{im} = \mathcal{S}\backslash \left\{ (M_m, K_i), (M_n, K_j) \right\} \cup \left\{ (M_n, K_i), (M_m, K_j) \right\}$, where $K_i \in \mathcal{S}_{jn}^{im}(M_n), K_j \in \mathcal{S}_{jn}^{im}(M_m), K_i \notin \mathcal{S}_{jn}^{im}(M_m),$ and $K_j \notin \mathcal{S}_{jn}^{im}(M_n)$.
\end{definition}

Note that not all swap operations are approved, considering the players' preferences. To elucidate the conditions for approval, we introduce the definition of a swap-blocking pair:
\begin{definition}[Swap-Blocking Pair] \label{def:swap-blocking_pair}
	$(K_i, K_j)$ is a pair in a given matching $\mathcal{S}$. Suppose there exists $M_m \in \mathcal{S}(K_i)$ and $M_n \in \mathcal{S}(K_j)$ such that:
	\begin{itemize}
		\item $\forall x \in \left\{ K_i, K_j, M_m, M_n \right\}, \left( \mathcal{S}_{jn}^{im}(x), \mathcal{S}_{jn}^{im} \right) \succcurlyeq_x \left( \mathcal{S}(x), \mathcal{S} \right)$,
		\item $\exists x \in \left\{ K_i, K_j, M_m, M_n \right\}, \left( \mathcal{S}_{jn}^{im}(x), \mathcal{S}_{jn}^{im} \right) \succ_x \left( \mathcal{S}(x), \mathcal{S} \right)$,
	\end{itemize}
	then the swap matching $\mathcal{S}_{jn}^{im}(x)$ is approved, and the pair $(K_i, K_j)$ is considered a swap-blocking pair in $\mathcal{S}$.
\end{definition}

Following multiple approved swap operations, the matching among the players can reach a two-sided exchange stable status, defined as follows:
\begin{definition}[Two-Sided Exchange Stable] \label{def:two-sided_exchange_stable}
	The matching $\mathcal{S}$ is considered two-sided exchange stable if none of the pairs $(K_i, K_j),~\forall i, j$ in $\mathcal{S}$ form a swap-blocking pair.
\end{definition}

\subsection{Overall TriMSM Algorithm}
With the definitions provided above, we introduce the overall TriMSM algorithm, which consists of two phases:

\subsubsection{Initialization Phase}
In this phase, we establish the initial matching between UEs and UBSs using the received power-based selection (RECP) method \cite{ngoTotalEnergyEfficiency2018} as the criterion for selecting UBSs for UEs. For each UE, the UBSs are ranked in ascending order according to the RECP criterion. Then, we select the top $\delta$\% UBSs to be matched with this UE. If the number of selected UBSs exceeds $L$, we only consider the top $L$ UBSs, taking into account constraint \eqref{P:st_UE_connect}. During this initial process for each UE, if one UBS is already matched with a number of UEs equal to $N$, indicating it's fully loaded, it is no longer available for matching with additional UEs, in compliance with constraint \eqref{P:st_BS_connect}.

\subsubsection{Swap Matching Phase}
In this phase, we first identify all possible UE pairs. For each pair, two UEs are selected to exchange their matched UBSs. The edge cloud then checks whether this swap operation would lead to a swap-blocking pair. If the swap operation is approved, the swap-blocking pair is removed after the swap is completed. This process continues until there are no more swap-blocking pairs in the matching, indicating that the matching is two-sided exchange stable.

\begin{algorithm}[!t]
	\caption{TriMSM algorithm for solving $\mathcal{P}_u$}
	\label{alg:TriMSM}
	\KwIn{Problem $\mathcal{P}_u$, RECP parameter $\delta$.}
	\KwOut{Optimal solution $\mathbf{S}^\ast, \mathbf{A}^\ast, \mathbf{P}^\ast$, and value $\mathrm{EE}^\ast$.}
	\textbf{Initialization Phase:} \\
	\For{each UE $K_i \in \mathcal{K}$}{
		Select the top $\delta$\% UBSs based on the RECP criterion subject to constraints \eqref{P:st_UE_connect} and \eqref{P:st_BS_connect}\;
	}
	Initialize the optimal $\mathbf{S}^\ast$ and $\mathbf{A}^\ast$ with the initial matching, and then initialize $\mathrm{EE}^\ast$ based on $\mathbf{S}^\ast$ and $\mathbf{A}^\ast$ using \algref{alg:SLMDB}\;
	\textbf{Swap Matching Phase:} \\
	Identify all possible pairs $\left( K_i, K_j \right)$ where $K_i, K_j \in \mathcal{K} \cup \emptyset $\;
	\While{there exists swap-blocking pair}{
		\For{each pair $\left( K_i, K_j \right)$}{
			\eIf{$\left( K_i, K_j \right) $ forms a swap-blocking pair along with $M_m \in \mathcal{S}(K_i), M_n \in \mathcal{S}(K_j)$}{
				Update the matching as $\mathcal{S} = \mathcal{S}^{im}_{jn}$, and then update the optimal $\mathbf{S}^\ast$ and $\mathbf{A}^\ast$ based on the new matching\;
				Update the optimal $\mathrm{EE}^\ast$ using $\mathbf{S}^\ast$ and $\mathbf{A}^\ast$ through the application of \algref{alg:SLMDB}\;
			}{
				Move to the next pair\;
			}
		}
	}
\end{algorithm}

The comprehensive description of the TriMSM algorithm is provided in \algref{alg:TriMSM}. The effectiveness and stability of the proposed TriMSM algorithm can be readily verified; detailed proofs can be referenced in \cite{yuFullydecoupledRadioAccess2023,diSubChannelAssignmentPower2016}.

\subsection{Three Low Complexity Alternative Power Control} \label{sec:low_complexity_power_control}
The complexity of \algref{alg:TriMSM}, as demonstrated by the theoretical and simulation results in Table~\ref{tab:complexity} and Section~\ref{sec:simulation_results}, limits its applicability to large-scale FD-RAN deployments.
To address this issue, we replace the computationally intensive optimal power control algorithm in the main loop of \algref{alg:TriMSM} with low-complexity heuristic methods. Once the user association converges, we apply the optimal power control algorithm in the final iteration to further enhance performance.
The variants are referred to as TriMSM+FiPC, TriMSM+QoPC, and TriMSM+EIPC, respectively.
This hybrid approach strikes a favorable balance between performance and computational efficiency, achieving results close to the optimal solution while significantly reducing overall complexity.

\subsubsection{Fixed Power Control (FiPC)}
The simplest approach is to employ a fixed power setting for all UEs. In this method, we set $P_k = P_{\max}$ for all UE $k \in \mathcal{K}$.

\subsubsection{QoS-constrained Power Control (QoPC)}
To ensure that the QoS constraints defined in \eqref{eq:st_convex_rate} are met for UEs to the greatest extent possible, we frame the following feasibility problem $\mathcal{P}_f$ to determine power control:
\begin{align}
	\mathcal{P}_f\text{:} ~ & \text{find}~\mathbf{P}                              \\
	\rm{s.t.}               & ~ \eqref{P:st_p}, \notag \eqref{eq:st_convex_rate},
\end{align}
which is convex and can be efficiently solved.

\subsubsection{Effective Channel Inversion Power Control (EIPC)}
In this approach, each UE's power is set proportionally to the inverse of the channel gain, aiming to achieve uniform received power at the UBSs. Taking into account the sleeping UBSs, we define the effective channel gain between UBS $m$ and UE $k$ as $\beta_{m,k} = \Vert S_{m,k} \mathbf{h}_{m, k} \Vert^2,~\forall m \in \mathcal{M}, k \in \mathcal{K}$. We represent the effective channel gains for UE $k$ as the vector $\pmb{\beta}_k = \left[ \beta_{1,k}, \beta_{2,k}, \cdots, \beta_{M,k} \right]^T,~\forall k \in \mathcal{K}$. The power of UE $k$ is then determined as follows:
\begin{align}
	P_k = \frac{\min_{k \in \mathcal{K}} \left\{ \Vert \pmb{\beta}_k \Vert^2 \right\} }{\Vert \pmb{\beta}_k \Vert^2} P_{\max},~\forall k \in \mathcal{K},
\end{align}
where the denominator represents the summation of channel gains from all active UBSs for UE $k$, considering that transmitted signals from UE $k$ would affect all active UBSs, while the numerator ensures that the power of UE $k$ does not exceed $P_{\max}$.

\subsection{Complexity Analysis} \label{sec:complexity_analysis}
\subsubsection{SLMDB Algorithm}
The SLMDB algorithm iteratively approximates the original problem $\mathcal{P}_l$ to obtain the optimal power $\mathbf{P}^{(n)}$ in the $n$-th iteration by solving the fractional problem $\mathcal{P}_l^\prime$. We denote the iteration number of this approximation procedure as $I_1$. The fractional problem $\mathcal{P}_l^{\prime}$ is addressed using Dinkelbach's algorithm, wherein the parametric convex problem $\mathcal{P}_l^{\prime\prime}$ is solved iteratively. The complexity of Dinkelbach's algorithm consists of two components: the iteration complexity and the per-iteration computation cost. We denote the iteration number of Dinkelbach's algorithm as $I_2$. In each per-iteration step, the convex problem $\mathcal{P}_l^{\prime\prime}$ is solved using the primal-dual interior-point method, with a computation complexity of $\mathcal{O}\left( K^3 \log \left( \epsilon^{-1} \right)  \right)$, where $\epsilon$ represents the accepted duality gap. Therefore, the total computation complexity of the SLMDB algorithm can be derived as $\mathcal{O}\left( I_1 I_2 K^3 \log \left( \epsilon^{-1} \right)  \right)$.

\subsubsection{Low Complexity Power Control Algorithms}
As shown in Section \ref{sec:low_complexity_power_control}, FiPC and EIPC have explicit expressions, resulting in a computational complexity of $\mathcal{O}\left( 1 \right)$. In the case of QoPC, the overall complexity is attributed to solving the convex problem $\mathcal{P}_f$, which can be efficiently addressed using the primal-dual interior-point method with a complexity of $\mathcal{O}\left( K^3 \log \left( \epsilon^{-1} \right)  \right)$, $\epsilon$, where $\epsilon$ represents the accepted duality gap.

\subsubsection{TriMSM Algorithm}
The computational complexity of the many-to-many swap matching algorithm can be attributed to the initialization phase and the swap matching phase. In the initialization phase, complexity primarily arises from obtaining power using the EIPC method, which is $\mathcal{O}\left( 1 \right) $. Additionally, sorting $M$ RSRP values for $K$ UEs to select the top UBSs introduces complexity, with an average complexity of $\mathcal{O}(M^2 K)$.
During each iteration of the swap matching phase, considering the possibility of empty sets, each UE has $L+1$ potential matching players. With $K$ UEs, there are $\binom{K+1}{2}$ possible pairs denoted as $\left( K_i, K_j \right)$. Hence, there can be at most $\binom{K+1}{2} \left( L+1 \right)^2$ potential swap operations. Let $I_3$ represent the total number of iterations, then the total number of swap operations can be expressed as $I_3 \binom{K+1}{2} \left( L+1 \right)^2$.
In each swap operation, determining the power is essential. For the original TriMSM algorithm, the SLMDB algorithm is employed to calculate the power control, with a complexity of $\mathcal{O}\left( I_1 I_2 K^3 \log \left( \epsilon^{-1} \right)  \right)$. Therefore, the total computational complexity of the original TriMSM algorithm can be denoted as $\mathcal{O} \left( M^2 K + I_1 I_2 I_3 \left( L+1 \right)^2 K^4 \left( K+1 \right) \log \left( \epsilon^{-1} \right)/2 \right) $.
For the TriMSM algorithm combined with the three other low-complexity power control methods, with complexities of $\mathcal{O}\left( 1 \right) $ or $\mathcal{O}\left( K^3 \log \left( \epsilon^{-1} \right)  \right) $, we can evaluate their complexities in a similar manner.

Regarding the exhaustive search method, it necessitates the exploration of all possible combinations ($2^{MK}$) of associations between UEs and UBSs. Assuming it employs the SLMDB algorithm for power control, we can readily deduce its overall complexity based on the earlier analysis.

We summarize the complexity of all algorithms in Table \ref{tab:complexity}. Upon comparison, it becomes evident that our proposed TriMSM algorithm exhibits significantly lower complexity than the exhaustive search method. Furthermore, the inclusion of three low-complexity alternatives further diminishes the overall complexity, rendering the algorithm well-suited for large-scale FD-RAN deployments.

\begin{table}[htbp]
	\centering
	\caption{Complexity of Various Algorithms}
	\label{tab:complexity}
	\renewcommand\arraystretch{1.2}     \resizebox{\linewidth}{!}{         \begin{tabular}{c|c}
			\toprule
			\hline
			\textbf{Algorithm} & \textbf{Complexity}                                                                                                                                                      \\    \hline
			Exhaustive Search  & $\mathcal{O} \left( 2^{MK} I_1 I_2 K^3 \log\left( \epsilon^{-1} \right)\right) $                                                                                         \\
			Original TriMSM    & $\mathcal{O} \left( M^2 K + I_1 I_2 I_3 \left( L+1 \right)^2 K^4 \left( K+1 \right) \log \left( \epsilon^{-1} \right)/2 \right) $                                        \\
			TriMSM + FiPC      & $\mathcal{O} \left( M^2 K + I_3 \left( L+1 \right)^2 K \left( K+1 \right) / 2 + I_1 I_2 K^3 \log\left( \epsilon^{-1} \right) \right) $                                   \\
			TriMSM + QoPC      & $\mathcal{O} \left( M^2 K + I_3 \left( L+1 \right)^2 K^4 \left( K+1 \right) \log \left( \epsilon^{-1} \right)/2 + I_1 I_2 K^3 \log\left( \epsilon^{-1} \right) \right) $ \\
			TriMSM + EIPC      & $\mathcal{O} \left( M^2 K + I_3 \left( L+1 \right)^2 K \left( K+1 \right) / 2 + I_1 I_2 K^3 \log\left( \epsilon^{-1} \right) \right) $                                   \\ \hline
			\bottomrule
		\end{tabular}
	}
\end{table}

\section[Simulation Results]{Simulation Results\footnote{Note that some missing data points in our simulation results are due to the absence of feasible solutions to problem $\mathcal{P}$.}} \label{sec:simulation_results}

In this section, we conduct comprehensive simulations to illustrate the energy efficiency advantages of FD-RAN and our proposed algorithms.
\subsection{Simulation Setups}
\begin{table}[!t]
	\centering
	\caption{Simulation Parameters \cite{zhaoFullydecoupledRadioAccess2023,debaillieFlexibleFutureProofPower2015,yuFullydecoupledRadioAccess2023,limEnergyefficientBestselectRelaying2012,fioraniModelingEnergyPerformance2016,lopez-perezSurvey5GRadio2022,dessetFlexiblePowerModeling2012,vanchienJointPowerAllocation2020,auerHowMuchEnergy2011}}
	\label{tab:simulation_parameters}
	\renewcommand\arraystretch{1.1} 
	\resizebox{\linewidth}{!}{ 
		\begin{tabular}{c|c|c|c|c|c}
			\toprule
			\hline
			\textbf{Parameter}            & \textbf{Value}  & \textbf{Parameter}              & \textbf{Value} & \textbf{Parameter}                       & \textbf{Value} \\ \hline
			$L$                           & 3               & $\tau_c$                        & 190            & $\tau_p$                                 & 10             \\
			$P_k^p$, $P_{\max}$           & 100 mW          & $\sigma^2$                      & $-94$ dBm      & $R_{k, \min}$                            & 20 Mbps        \\
			$N_{\mathrm{ref}}$            & 1               & $B_{\mathrm{ref}}$              & 20 MHz         & $Q_{\mathrm{ref}}$                       & 24 bit         \\
			$Se_{\mathrm{ref}}$           & 6 bps/Hz        & $Ld_{\mathrm{ref}}$             & 100 \%         & $St_{\mathrm{ref}}$                      & 1              \\
			$N (N_{\mathrm{act}})$        & 5               & $B (B_{\mathrm{act}})$          & 20 MHz         & $Q_{\mathrm{act}}$                       & 24 bit         \\
			$Se_{\mathrm{act}}$           & 6 bps/Hz        & $Ld_{\mathrm{act}}$             & 100 \%         & $St_{\mathrm{act}}$                      & 1              \\
			$\nu_p$                       & $6 \times 10^5$ & $N_i^s$                         & 1              & $\sigma_{\mathrm{MS}}$                   & 0.1            \\
			$\sigma_{\mathrm{DC}}$        & 0.05            & $\sigma_{\mathrm{CO}}$          & 0              & $\psi_d$                                 & 80\%           \\
			$\eta_s$                      & 10\%            & $P_i^{\mathrm{FH_{fix}}}$       & 0.825 W        & $\Delta_i^{\mathrm{FH}_{\mathrm{trf}}} $ & 0.25 W/Gbps    \\
			$\kappa$                      & 1               & $\zeta$                         & 2              & $\lambda$                                & 5              \\
			$\xi$                         & 2               & $\varrho$                       & 2              & $\sigma_{\mathrm{CO}}^{\prime}$          & 0.1            \\
			$\Delta_k^{\mathrm{UE_{pa}}}$ & 2.6             & $P_k^{\mathrm{UE}_\mathrm{cp}}$ & 1.31 W         & $R_{k, \mathrm{ref}}$                    & 40   Mbps      \\ \hline
			\bottomrule
		\end{tabular}
	}
\end{table}
The considered uplink FD-RAN scenario aligns with the network model outlined in Section~\ref{sec:network_model}. Three types of BSs are randomly distributed within a 500m $\times$ 500m square using the wrap-around method, and channel model described in \cite{zhaoFullydecoupledRadioAccess2023} is employed.
As our analysis focuses on the variable power consumption of uplink data transmission, the constant power of the CBS $P^{\text{CBS}}$ is excluded from the simulation.
The reference power tables for the RF unit and the BBU are provided in \cite[Table \uppercase\expandafter{\romannumeral3} and Table \uppercase\expandafter{\romannumeral4}]{debaillieFlexibleFutureProofPower2015}, with their corresponding scaling factors detailed in \cite[Table \uppercase\expandafter{\romannumeral3} and Table \uppercase\expandafter{\romannumeral4}]{debaillieFlexibleFutureProofPower2015}. The remaining simulation parameters are summarized in Table~\ref{tab:simulation_parameters}.

We consider the following benchmarks regarding the algorithms: the three-step access procedure (TSAP), where the neighborhood UBSs are defined as those with 30\% of the maximum large-scale fading \cite{bjornsonScalableCellFreeMassive2020}; the RECP with $\delta=95$ \cite{ngoTotalEnergyEfficiency2018}; and the largest-large-scale-fading-based selection (LLSF) \cite{pintoantonioliEnergyEfficiencyCellfree2022a}. Besides, the no BS sleeping version of TriMSM with EIPC (NoS-TriMSM) is evaluated. Notably, these algorithms are employed alongside \algref{alg:SLMDB} to establish benchmarks.
For the benchmarks of architectures, we exclusively focus on the uplink power aspects and employ identical configurations to those of FD-RAN, highlighting the differing characteristics.
The cellular network uses single-connectivity, lacks centralized gain, and requires cooling at the BSs. The total antenna count matches that of the FD-RAN, yet this network consists of only 4 distributed BSs.
The small-cell network, it also employs single-connectivity and lacks centralized gain.
The cell-free network, similar in lacking centralized gain, has two implementations: full connection (F-Cell-Free) and UE-centric (UC-Cell-Free). F-Cell-Free establishes full associations between all UEs and BSs, while UC-Cell-Free mirrors associations as in FD-RAN.
Additionally, UC-Cell-Free considers coupled uplink and downlink. We assume that idle UBSs can enter sleep mode with a probability from 0 to 1, indicating the impact of downlink transmission on UBSs. This defines the UC-Cell-Free Region, depicted with green shading, where a solid green line represents a probability of 0.5. This illustration is shown in \figref{fig:2_EE_vs_architecture} and \figref{fig:6b_varyingtraffic_vs_arc}.

\subsection{UBS Sleeping and UE Association}
\begin{figure}[htbp]
	\centering
	\centerline{\includegraphics[width=0.7\linewidth]{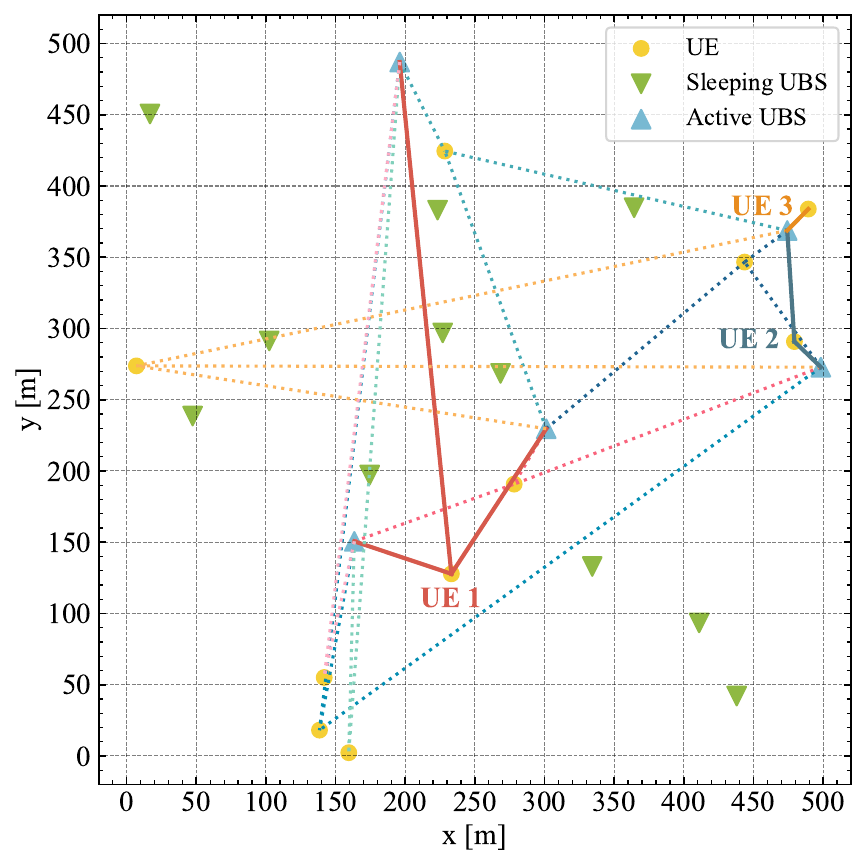}}
	\captionsetup{width=.7\linewidth}
	\caption{UBS sleeping and UE associations in FD-RAN ($M=16$, $K=10$).}
	\label{fig:1_UBS_UE_association}
\end{figure}
We present a visual representation of UE association and UBS sleeping in \figref{fig:1_UBS_UE_association}. Facilitated by the proposed TriMSM algorithm, a cluster of UBSs is assigned to serve each UE, strategically placing underutilized UBSs into sleep mode to conserve energy. Notably, cluster size varies and is capped at $L$, aiming at maximizing energy efficiency.

\subsection{Energy Efficiency versus Different Architectures}
\begin{figure*}[!t]
	\centering
	\begin{minipage}{\linewidth}
		\centering
		\captionsetup{width=.33\linewidth}
		\subfloat[CDF curves of energy efficiency ($M=16$, $K=5$).\label{fig:2a_EE_CDF}]{
			\includegraphics[width=0.33\linewidth]{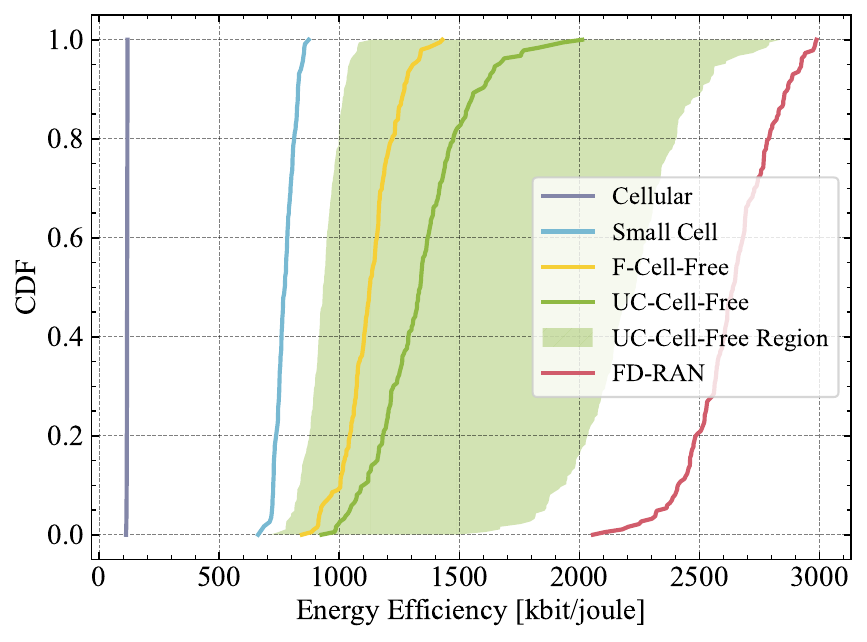}
		}
		\subfloat[Average energy efficiency versus different number of UEs ($M=16$). \label{fig:2b_EEvsK}]{
			\includegraphics[width=0.33\linewidth]{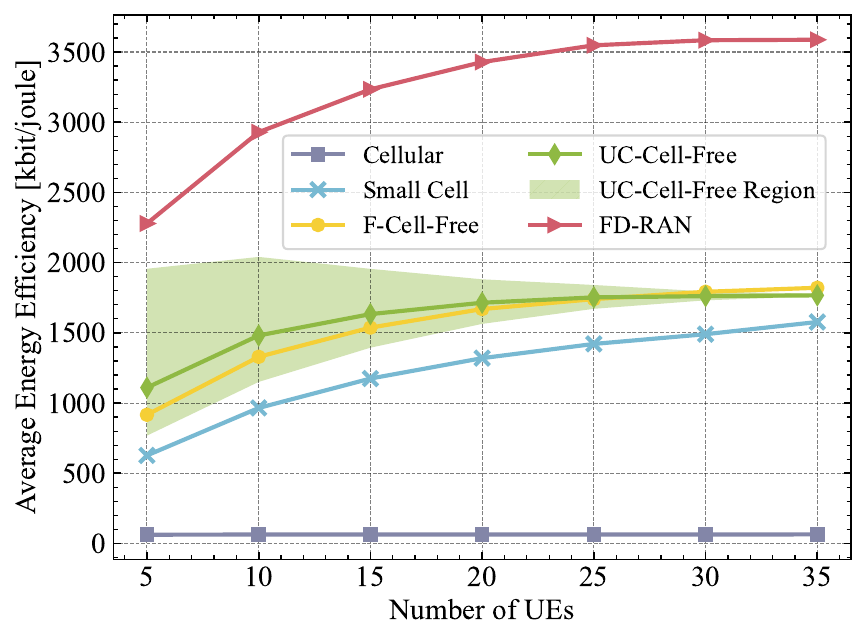}
		}
		\subfloat[Average energy efficiency versus different number of UBSs ($K=10$). \label{fig:2d_EEvsM}]{
			\includegraphics[width=0.33\linewidth]{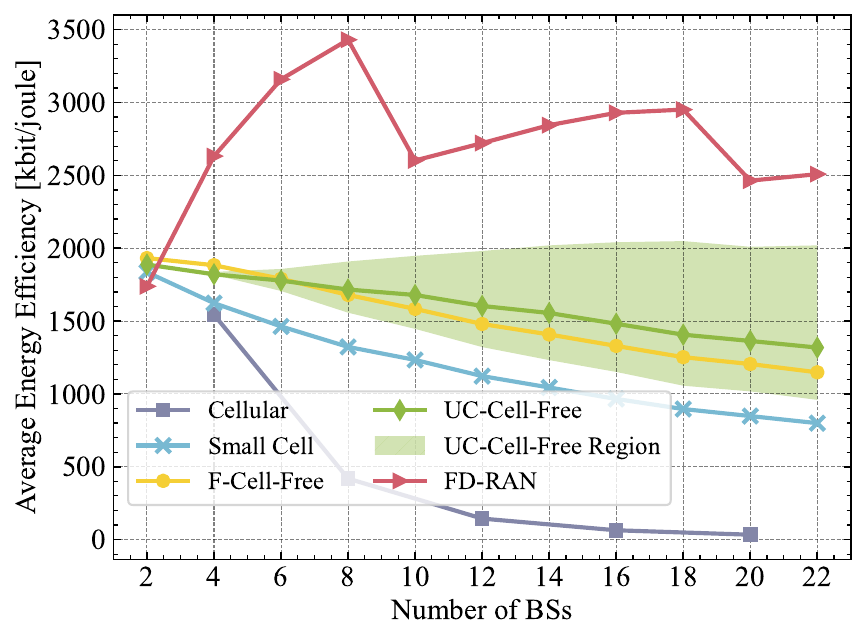}
		}
		\captionsetup{width=\linewidth}
		\caption{Energy efficiency versus different architectures  (using the TriMSM+EIPC algorithm).}
		\label{fig:2_EE_vs_architecture}
	\end{minipage}
\end{figure*}
\begin{figure*}[!t]
	\centering
	\captionsetup{width=.35\linewidth}
	\begin{minipage}{0.35\linewidth}
		\centering
		\includegraphics[width=\linewidth]{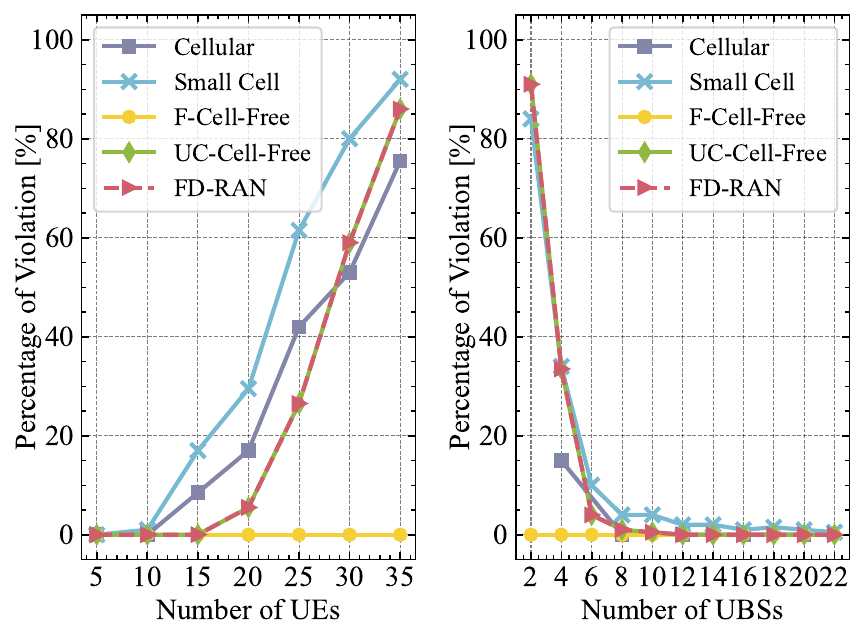}
		\caption{QoS violation percentage versus different architectures (with different number of UEs ($M=16$) and UBSs ($K=10$)).}
		\label{fig:2ce_violation_rate_vs_KM}
	\end{minipage}
	\begin{minipage}{0.35\linewidth}
		\centering
		\includegraphics[width=\linewidth]{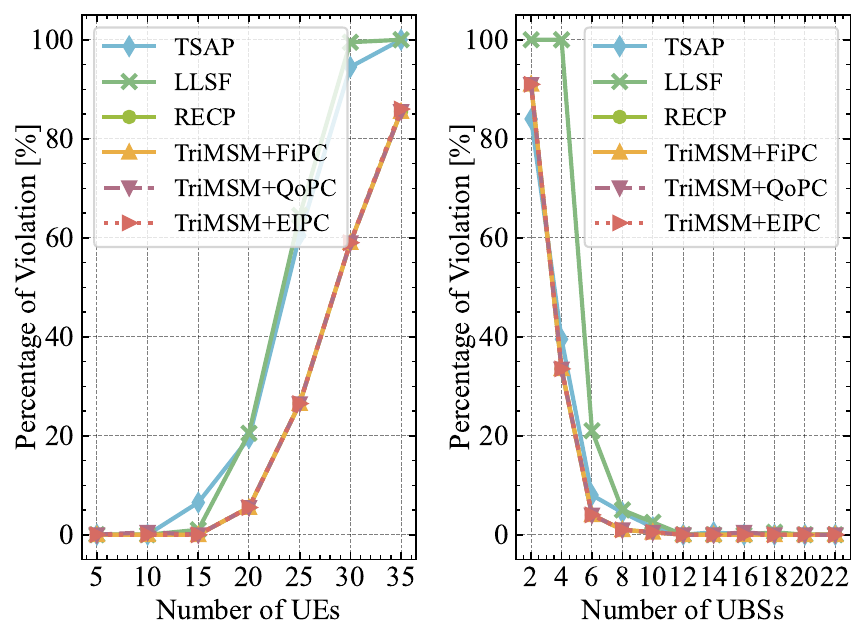}
		\caption{QoS violation percentage versus different algorithms (with different number of UEs ($M=16$) and UBSs ($K=10$)).}
		\label{fig:4ce_violation_rate_vs_KM}
	\end{minipage}
\end{figure*}

\figref{fig:2_EE_vs_architecture} illustrates the energy efficiency versus different architectures.
The cumulative distribution function (CDF) curves for energy efficiency are displayed in \figref{fig:2a_EE_CDF}.
FD-RAN consistently outperforms all other existing architectures, exhibiting energy efficiency values 22.7, 3.40, 2.34, and 1.97 times higher than those of cellular, small cell, F-Cell-Free, and UC-Cell-Free, respectively. Notably, even in the best-case of UC-Cell-Free, the FD-RAN architecture maintains a significant 18.9\% advantage in energy efficiency.
\figref{fig:2b_EEvsK} presents the average energy efficiency with varying numbers of UEs. As the number of UEs increases, the energy efficiency initially rises but gradually reaches a saturation point in most cases. This is primarily due to the almost full utilization of UBSs and the saturation of UE rate caused by limited network capacity.
Notably, FD-RAN consistently outperforms the other architectures, with its superiority becoming even more pronounced as the number of UEs increases. This advantage stems from its centralized gains of BBUs.
\figref{fig:2d_EEvsM} illustrates the average energy efficiency versus the number of UBSs\footnote{In the case of cellular architecture, the number of BSs is fixed at 4. Consequently, the line depicted in the figure represents changes in energy efficiency as the number of antennas varies while maintaining the total number of antennas equal to that of FD-RAN. As a result, we can only represent cases that are multiples of 4.}. Generally, as the number of BSs increases, the energy efficiency of most architectures decreases.
However, in UC-Cell-Free and FD-RAN, the introduction of BS sleeping helps effectively manage the rising power consumption of additional BSs, resulting in higher energy efficiency. This effect of BS sleeping is evident as shown in the UC-Cell-Free Region.
The sharp drop of energy efficiency in FD-RAN will be explained in Section~\ref{sec:EE_vs_alg}.
Furthermore, FD-RAN consistently outperforms other architectures except in cases with 2 UBSs ($M=2$). This deviation can be attributed to the additional cooling energy required in centralized BBU but the diminished centralized gain, only when $M$ is quite small.
Based on the findings, FD-RAN's remarkable energy efficiency stems from a flexible BS sleeping mechanism, enabled by its fully decoupled architecture, multi-connectivity, and centralized gain.

\begin{figure*}[!t]
	\centering
	\captionsetup{width=.32\linewidth}
	\begin{minipage}{0.32\linewidth}
		\centering
		\includegraphics[width=\linewidth]{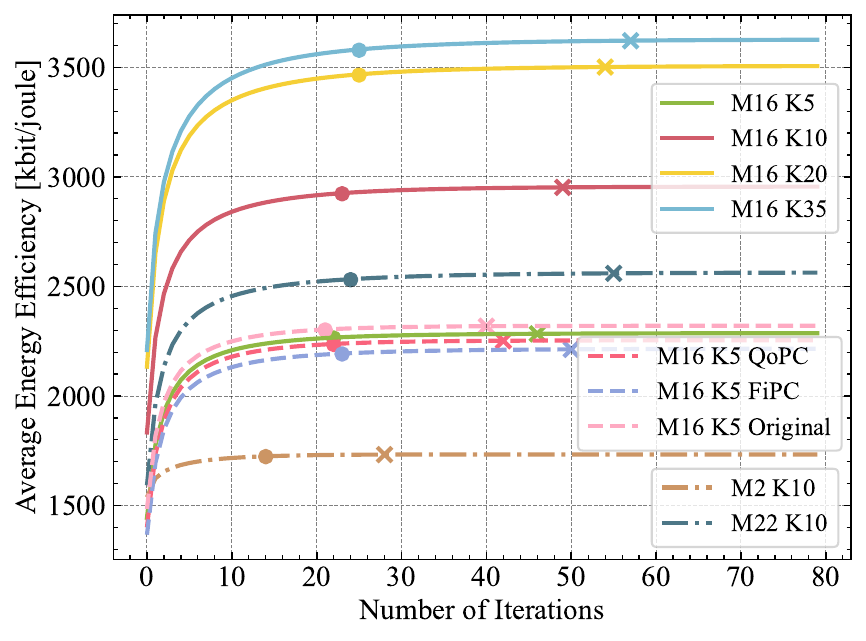}
		\caption{Convergence of SLMDB algorithm versus different algorithms, number of UEs and UBSs.}
		\label{fig:3a_SLMDB_Convergence}
	\end{minipage}
	\begin{minipage}{0.33\linewidth}
		\centering
		\includegraphics[width=\linewidth]{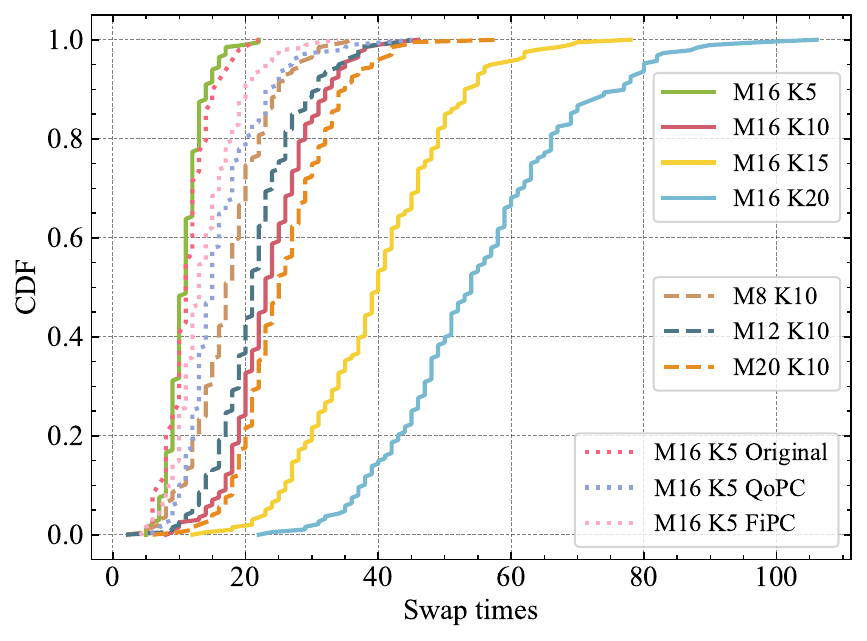}
		\caption{Swap times of TriMSM algorithm versus different algorithms, number of UEs and UBSs.}
		\label{fig:3b_swap_times}
	\end{minipage}
	\begin{minipage}{0.33\linewidth}
		\centering
		\includegraphics[width=\linewidth]{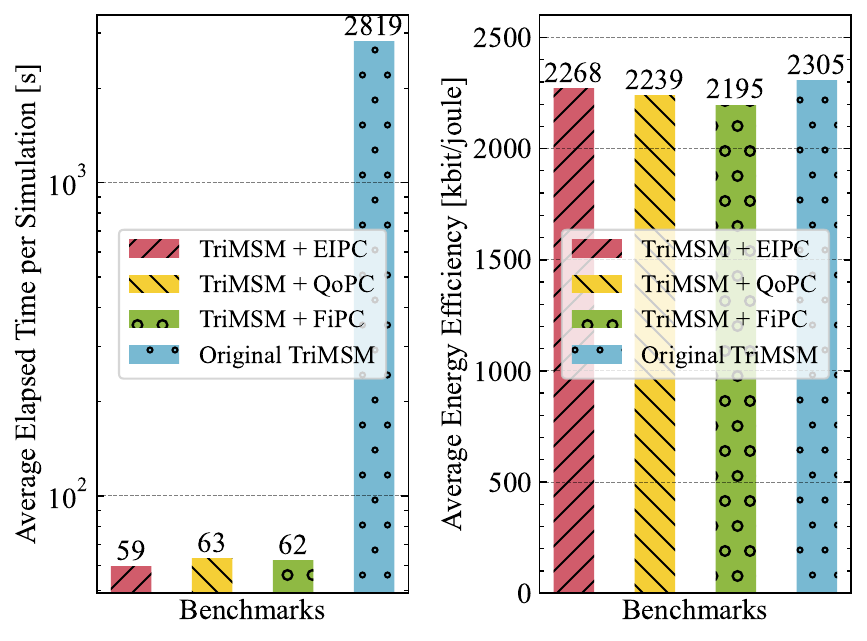}
		\caption{Average elapsed time (left) and average energy efficiency (right) comparison of TriMSM algorithms ($M=16$, $K=5$).}
		\label{fig:3c_run_time_performance}
	\end{minipage}
\end{figure*}
\begin{figure*}[!t]
	\centering
	\captionsetup{width=.33\linewidth}
	\begin{minipage}{\linewidth}
		\centering
		\captionsetup{width=.32\linewidth}
		\subfloat[CDF curves of energy efficiency ($M=16$, $K=5$).\label{fig:4a_EE_CDF}]{
			\includegraphics[width=0.33\linewidth]{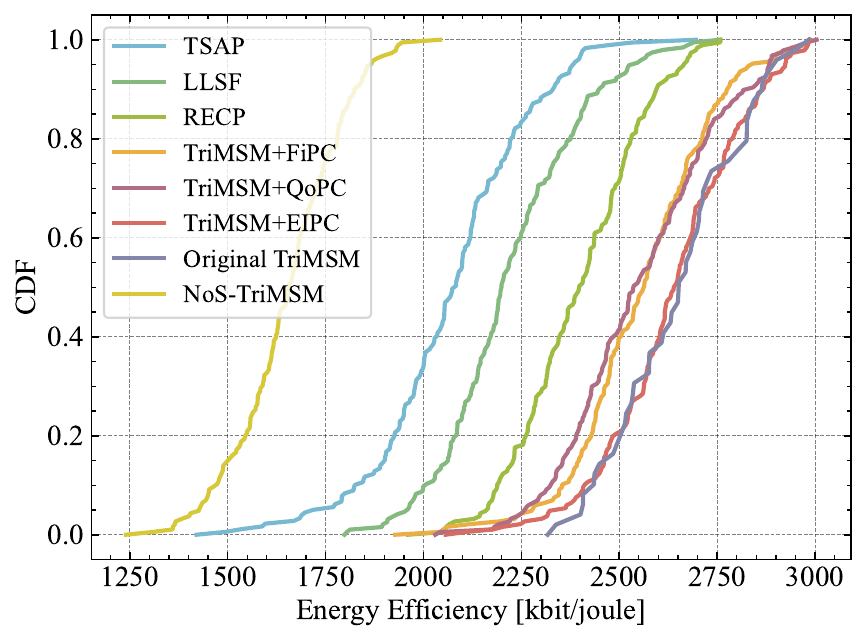}
		}
		\subfloat[Average energy efficiency versus different number of UEs ($M=16$).\label{fig:4b_EEvsK}]{
			\includegraphics[width=0.33\linewidth]{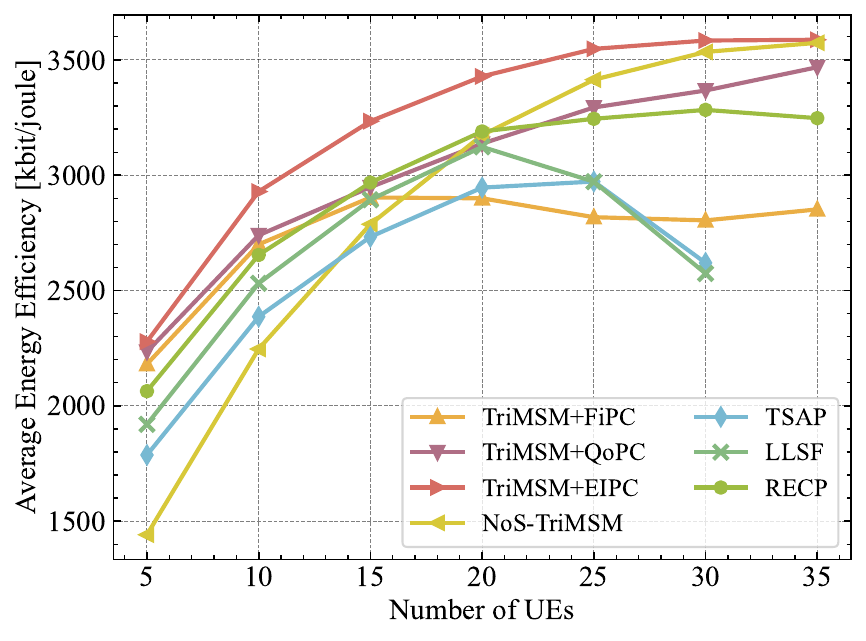}
		}
		\subfloat[Average energy efficiency versus different number of UBSs ($K=10$).\label{fig:4d_EEvsM}]{
			\includegraphics[width=0.33\linewidth]{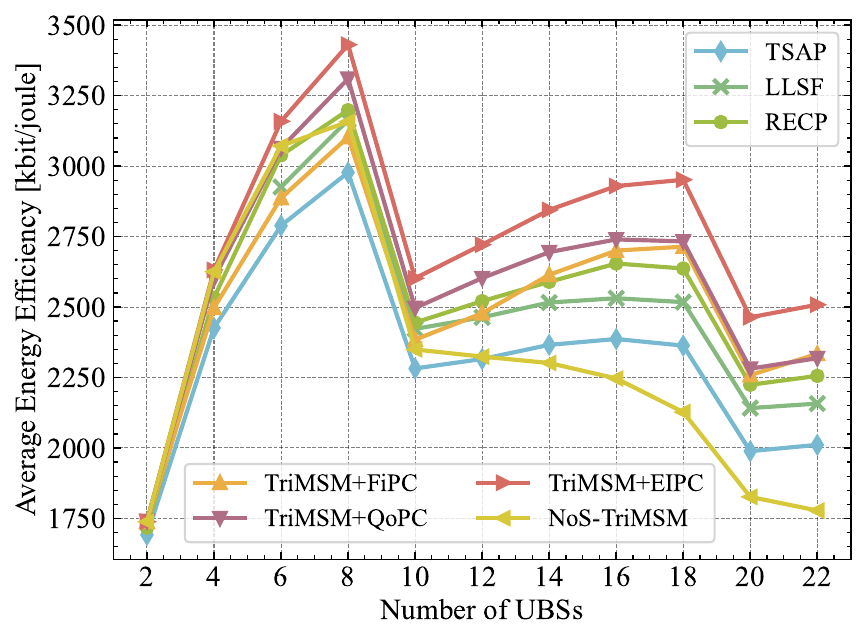}
		}
		\captionsetup{width=\linewidth}
		\caption{Energy efficiency versus different algorithms.}
		\label{fig:4_EE_vs_alg}
	\end{minipage}
\end{figure*}

\figref{fig:2ce_violation_rate_vs_KM} illustrates the QoS violations across different architectures. While F-Cell-Free delivers a commendable performance, it neglects to address the constraints \eqref{P:st_BS_connect} and \eqref{P:st_UE_connect}, rendering it impractical.
FD-RAN consistently provides superior QoS guarantees in most cases, while UC-Cell-Free benefits from mirrored UE association. However, in resource-constrained settings (e.g., $K=30$ and $35$ for $M=16$, and $M=4$ for $K=10$), cellular networks perform better. This is because, compared to single-connectivity in cellular, multi-connectivity in such scenario yields less gain and can result in uneven resource distribution when maximizing energy efficiency.

\subsection{Effectiveness of Proposed Algorithms}
\figref{fig:3a_SLMDB_Convergence} illustrates the convergence curves of the SLMDB algorithm across different power control algorithms, numbers of UEs and UBSs. Solid and cross markers denote energy efficiency variation points of 1e-3 and 1e-4, respectively. The lines depict convergence behavior in TriMSM with EIPC unless stated otherwise.
Various power control algorithms in TriMSM show similar convergence patterns, with original TriMSM having the best and FiPC the slowest convergence.
Although SLMDB's convergence slows with more UEs or UBSs, the impact remains relatively minor. Typically, about 20 steps suffice to reach the 1e-3 point. Therefore, the SLMDB algorithm demonstrates rapid convergence and robustness across various scenarios.

\figref{fig:3b_swap_times} shows the CDF of swap times for the TriMSM algorithm across various power control algorithms, numbers of UEs, and UBSs. The lines depict convergence behavior in TriMSM with EIPC unless stated otherwise. Notably, TriMSM with EIPC and original TriMSM exhibit similar swap times, generally with the fewest swaps, while FiPC has more swaps, and QoPC has the most.
Despite the increasing UEs or UBSs, swap times in all cases stay within acceptable limits, averaging at most 55 swaps. Thus, \figref{fig:3b_swap_times} highlights manageable swap times for the TriMSM algorithm, even in large-scale network environments.

To evaluate computational time, we compare the original TriMSM algorithm with three low-complexity alternatives, employing only a single processing core\footnote{Notably, the TriMSM algorithm is inherently parallel, utilizing multiple MATLAB cores, which can significantly reduce running times.}.
As depicted in \figref{fig:3c_run_time_performance}, the results reveal that the original TriMSM algorithm suffers from notably high elapsed times, rendering it unsuitable for large-scale network applications, even when considering parallelization capabilities.
In contrast, the three low-complexity alternatives demonstrate significantly reduced running times with nearly comparable performance. Among them, EIPC emerges as the top performer, achieving a 47.5-fold reduction in running time with just a 1.61\% loss in performance. The other two alternatives also show lower but notably good performance.

\subsection{Energy Efficiency versus Different Algorithms} \label{sec:EE_vs_alg}
\begin{figure*}[!t]
	\centering
	\captionsetup{width=.33\linewidth}
	\begin{minipage}{.33\linewidth}
		\centering
		\includegraphics[width=\linewidth]{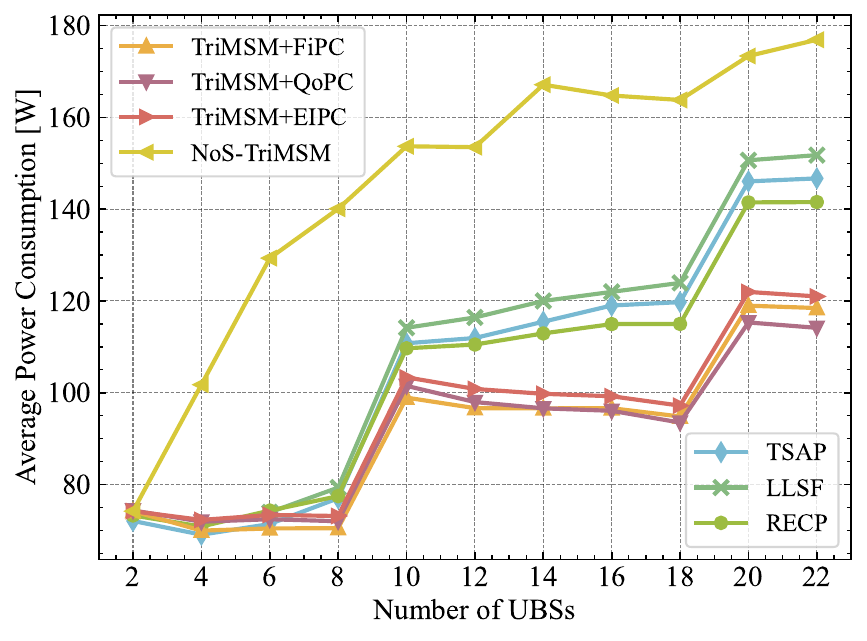}
		\caption{Average power consumption of algorithms versus different number of UBSs ($K=10$).}
		\label{fig:4g_PvsM}
	\end{minipage}
	\begin{minipage}{0.66\linewidth}
		\centering
		\subfloat[Different number of UEs ($M=16$).\label{fig:5a_activenum_vs_K}]{
			\includegraphics[width=.5\linewidth]{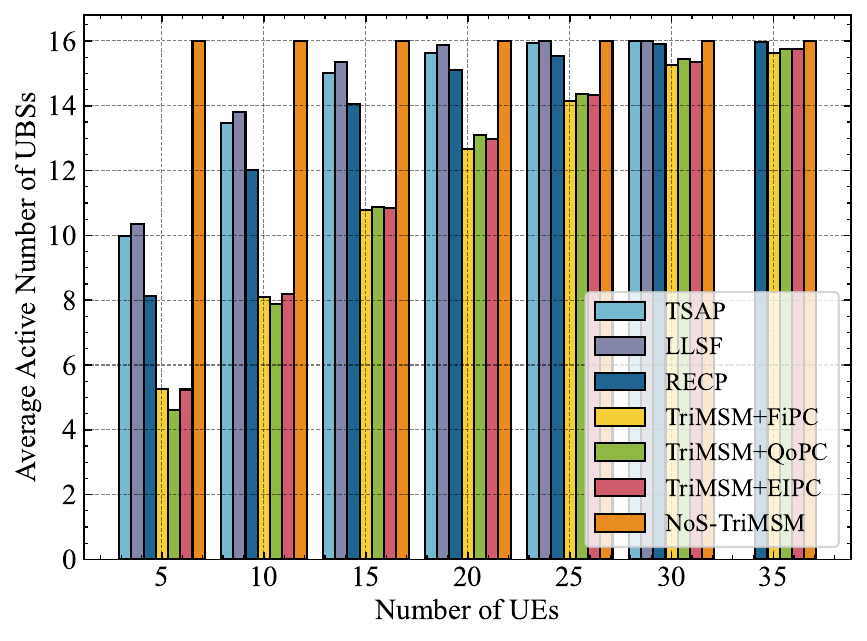}
		}
		\subfloat[Different number of UBSs ($K=10$).\label{fig:5b_activenum_vs_M}]{
			\includegraphics[width=.5\linewidth]{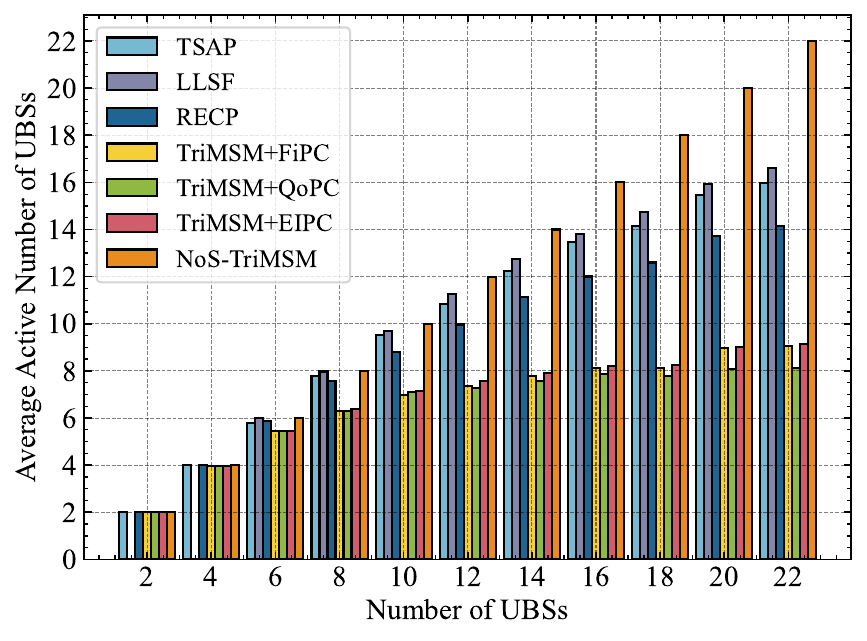}
		}
		\captionsetup{width=.66\linewidth}
		\caption{Average active number of UBSs versus different algorithms.}
		\label{fig:5_Ave_num_vs_alg}
	\end{minipage}
\end{figure*}
\begin{figure*}[!t]
	\centering
	\captionsetup{width=.35\linewidth}
	\begin{minipage}{\linewidth}
		\centering
		\subfloat[Comparison of different architectures (using the TriMSM+EIPC algorithm).\label{fig:6b_varyingtraffic_vs_arc}]{
			\includegraphics[width=0.35\linewidth]{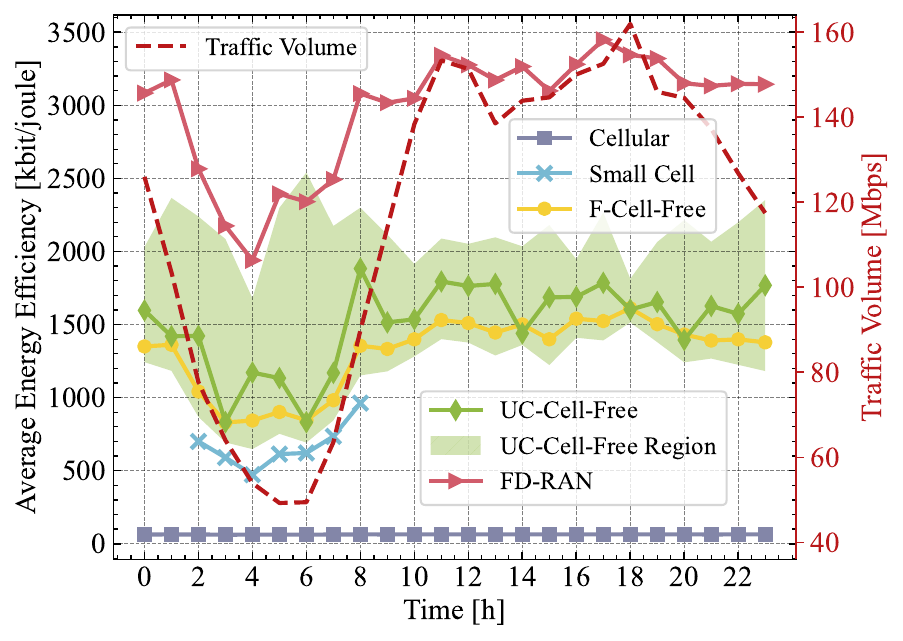}
		}
		\subfloat[Comparison of different algorithms.\label{fig:6a_varyingtraffic_vs_alg}]{
			\includegraphics[width=0.35\linewidth]{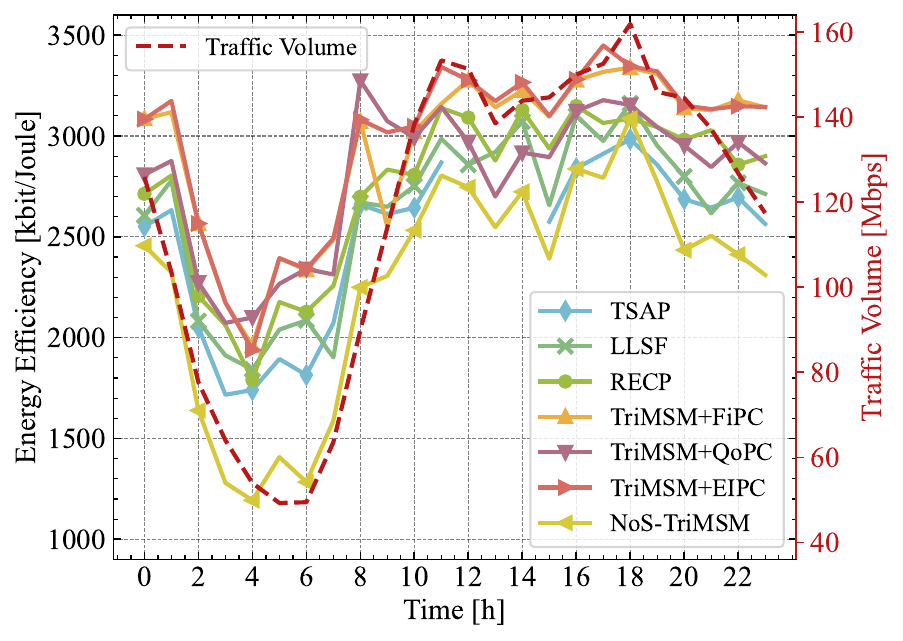}
		}
		\captionsetup{width=\linewidth}
		\caption{Energy efficiency versus real traffic.}
		\label{fig:6_varyingtraffic}
	\end{minipage}
\end{figure*}

\figref{fig:4_EE_vs_alg} illustrates the energy efficiency versus different algorithms.
\figref{fig:4a_EE_CDF} displays the CDF curves of energy efficiency, revealing that TriMSM algorithms exhibit the best performance, followed by peer algorithms, with the no-sleep algorithm performing the worst. The worst performing proposed algorithm achieves energy efficiency 6.60\% and 23.4\% higher than the best and worst peer algorithms, respectively. Notably, NoS-TriMSM exhibits the lowest efficiency, showing a 1.59-fold decrease in energy efficiency compared to TriMSM with sleeping, emphasizing the benefits of BS sleeping. Within the TriMSM algorithms, the original TriMSM demonstrates the best performance, while EIPC shows almost identical performance. FiPC and QoPC exhibit a minor performance gap.
The average energy efficiency concerning the number of UEs is illustrated in \figref{fig:4b_EEvsK}.
The energy efficiency of algorithms initially rises and then declines with an increasing number of UEs, except for the TriMSM algorithms. This trend is attributed to excessive and redundant UBS utilization in peer algorithms, evident in \figref{fig:5a_activenum_vs_K}.
TriMSM with EIPC consistently demonstrates the highest energy efficiency, with NoS-TriMSM gradually approaching it as the number of UEs increases due to nearly all UBS utilization.
However, TriMSM with FiPC and QoPC exhibits less satisfactory energy efficiency. As depicted in \figref{fig:5b_activenum_vs_M}, the UBS utilization among TriMSM algorithms is similar, suggesting that the decline in performance stems from their suboptimal power control strategies.
In \figref{fig:4d_EEvsM}, the average energy efficiency concerning the number of UBSs is depicted. There is a significant decline in energy efficiency when the number of UBSs increases tenfold, linked to FD-RAN's power consumption, as evident in \figref{fig:4g_PvsM}. This phenomenon can be explained by \eqref{eq:power_EC_new}, where every $\lambda \zeta = 10$ points experience a notable surge in power consumption.
The TriMSM serial algorithms exhibit the highest growth rates of energy efficiency within the small intervals, with EIPC consistently being the most efficient choice. This superiority stems from effective BS sleeping and centralized gain, leading to reduced power consumption, as evidenced in \figref{fig:4g_PvsM}. Notably, TriMSM with FiPC exhibits subpar energy efficiency in heavy-load scenarios ($M < 14$), whereas NoS-TriMSM performs well under heavy loads ($M < 6$) but experiences declining performance as the load decreases.

\figref{fig:5_Ave_num_vs_alg} presents the average active number of UBSs across different algorithms. NoS-TriMSM engages all available UBSs, while TSAP, LLSF, and RECP successively utilize fewer UBSs. Conversely, our proposed TriMSM algorithms prioritize the sleeping of most UBSs. Notably, the most efficient TriMSM variant, EIPC, doesn't feature the fewest active UBSs. Therefore, optimizing energy efficiency should consider a balance between power consumption and achievable UE rates, rather than merely minimizing the number of active UBSs.
In \figref{fig:5a_activenum_vs_K}, as the number of UEs increases, more UBSs are utilized. Unlike peer algorithms constantly requiring numerous active UBSs, TriMSM algorithms can efficiently put more UBSs to sleep, dynamically adapting to the number of UEs.
Meanwhile, \figref{fig:5b_activenum_vs_M} illustrates a rise in active UBSs concerning UBS number. Notably, the growth rate of TriMSM algorithms is notably restrained compared to the near-linear increments seen in peer algorithms. This observation sheds light on their superior energy efficiency, as demonstrated in \figref{fig:4_EE_vs_alg}.

\figref{fig:4ce_violation_rate_vs_KM} depicts QoS violations across various algorithms.
LLSF registers the highest QoS violation rate, while TSAP shows fewer violations. Notably, both RECP and TriMSM algorithms share the same lowest percentage of QoS violations, since RECP serves as the foundational algorithm for TriMSM algorithms. This emphasizes the effectiveness of TriMSM in maintaining QoS.

\subsection{Energy Efficiency versus Real Traffic}
We use the dataset from \cite{zhangDeepTransferLearning2019}, comprising real-world spatiotemporal traffic data. The dataset consists of 100$\times$100 cells over 62 days, recorded at hourly intervals. To assess FD-RAN and our proposed algorithms' real-world performance while maintaining generality, we aggregate the traffic from these cells into 5$\times$5 regions, partitioning the simulation area accordingly. If a region's traffic constitutes less than 20\% of the total, its traffic is set to 0. Each region is represented by a single UE placed at its center. Moreover, the total traffic across regions fluctuates over time, shown by the dotted line in \figref{fig:6_varyingtraffic}, with maximum traffic scaled to 160 Mbps.

We compare the energy efficiency of different network architectures and algorithms using real-world traffic data in \figref{fig:6b_varyingtraffic_vs_arc} and \figref{fig:6a_varyingtraffic_vs_alg}, respectively.
As shown in Fig. 9a, the energy efficiency fluctuates in response to real traffic variations in all network architectures, however, FD-RAN consistently achieves superior performance compared with other networks.
For comparisons of algorithms, the proposed three TriMSM variants show robust adaptability to real traffic, and EIPC and QoPC generally outperform FiPC. Even under low-traffic conditions, they maintain high energy efficiency. In contrast, NoS-TriMSM, despite its responsiveness to traffic volume, achieves the lowest energy efficiency due to the lack of BS sleeping.

\section{Conclusion} \label{sec:conclusion}
In this paper, we have studied adaptive BS sleeping and resource allocation in a green uplink FD-RAN. We have developed a holistic power consumption model for FD-RAN and defined a maximizing energy efficiency problem. Subsequently, we have decomposed this problem into a power control problem and a joint UE association and BS sleeping problem, which have been tackled by the successive lower-bound maximization-based Dinkelbach's algorithm and the modified many-to-many matching algorithm, along with low-complexity realizations, respectively.
The extensive simulation results have demonstrated the improved energy efficiency of FD-RAN and the effectiveness of the proposed algorithms. These outcomes reveal that the predominant sources of energy efficiency gains in FD-RAN stem from a flexible BS sleeping mechanism enabled by the fully decoupled network architecture, multi-connectivity, and centralized gains.
For future work, we will explore the green potential of downlink FD-RAN based on location-mapping transmission, considering the challenges of delay and feedback overhead.

\begin{appendices}
	\section{Proof of \lemmref{lemma:UBS_power_rate_load}} \label{app:proof_lemma_UBS_power_rate_load}
	\begin{proof}
		Considering that the number of antennas, bandwidth, quantization, spectral efficiency, and streams are typically fixed in the real world, we can express the power consumption of UBSs as a combination of fixed and varying parts, as follows\footnote{Here, we denote $ \sum_{m \in \mathcal{M}} A_m $ as $\sum_{m \in \mathcal{M_A}} $ for brevity.}:
		\begin{align}
			\sum_{m \in \mathcal{M_A}} P_m^{\mathrm{BS}} & = \sum_{m \in \mathcal{M_A}} P_m^{\mathrm{BS_{fix}}} + \sum_{m \in \mathcal{M_A}} \sum_{j \in \mathcal{J}_\mathrm{BBU}} P_m^{\mathrm{BBU}_{j, \mathrm{fix}}} \notag                            \\
			                                             & + \sum_{m \in \mathcal{M_A}} \sum_{j \in \mathcal{J}_\mathrm{BBU}} P_m^{\mathrm{BBU}_{j, \mathrm{ref}}} \left( \frac{Ld_m}{Ld_m^{\mathrm{ref}}} \right)^{s_m^{j, Ld}}, \label{eq:appendix_A_1}
		\end{align}
		where the first two terms represent the fixed energy of the BSs (excluding BBUs) and BBUs, respectively. The last term corresponds to the varying energy of BBUs, which depends on the load and can be further expanded as:
		\begin{align}
			\sum_{m \in \mathcal{M_A}} \left( \sum_{j \in \mathcal{J^{\prime}}_\mathrm{BBU}} P_m^{\mathrm{BBU}_{j, \mathrm{ref}}} + \sum_{j \in \mathcal{J^{\prime\prime}}_\mathrm{BBU}} P_m^{\mathrm{BBU}_{j, \mathrm{ref}}} \frac{Ld_m^{s_m^{j, Ld}}}{{Ld_m^{\mathrm{ref}}}^{s_m^{j, Ld}}}  \right), \label{eq:appendix_A_2}
		\end{align}
		where $s_m^{j, Ld} = 0$ for $j \in \mathcal{J}_{\mathrm{BBU}}^{\prime}$, and $s_m^{j, Ld} = 1$ or $0.5$ for $j \in \mathcal{J}_{\mathrm{BBU}}^{\prime\prime}$. To simplify the expression, we set $s_m^{j, Ld} = 1$ for $j \in \mathcal{J}_{\mathrm{BBU}}^{\prime\prime}$, and the varying part of \eqref{eq:appendix_A_2} can be calculated as:
		\begin{align}
			                 & \sum_{m \in \mathcal{M_A}} \sum_{j \in \mathcal{J}^{\prime\prime}_{\mathrm{BBU}}} P_m^{\mathrm{BBU}_{j, \mathrm{ref}}} \frac{Ld_m}{{Ld_m^{\mathrm{ref}}}} \notag                  \\
			\overset{(a)}{=} & \sum_{j \in \mathcal{J}^{\prime\prime}_{\mathrm{BBU}}} P_m^{\mathrm{BBU}_{j, \mathrm{ref}}} \sum_{m \in \mathcal{M_A}} \frac{R^{\prime}_m}{{R^{\prime}_{m, \mathrm{ref}}}} \notag \\
			\overset{(b)}{=} & \sum_{k \in \mathcal{K}} \frac{R_k}{{R_{k, \mathrm{ref}}}} \sum_{j \in \mathcal{J}^{\prime\prime}_{\mathrm{BBU}}} P_m^{\mathrm{BBU}_{j, \mathrm{ref}}}, \label{eq:appendix_A_3}
		\end{align}
		where step $(a)$ swaps the order of summation and substitutes \eqref{eq:load_definition}, while step $(b)$ is obtained by utilizing \eqref{eq:relation_R_k_R_m} and swapping the order of summation. By combining equations \eqref{eq:appendix_A_1}-\eqref{eq:appendix_A_3}, and denoting the summation of all fixed energy as $\sum_{m \in \mathcal{M_A}} P_m^{\mathrm{BS}_\mathrm{fix}}$ and $\sum_{j \in \mathcal{J}^{\prime\prime}_{\mathrm{BBU}}} P_m^{\mathrm{BBU}_{j, \mathrm{ref}}}$ as $P_{\mathrm{trf}}$, we complete the proof.
	\end{proof}

	\section{Proof of \lemmref{lemm:parametric_problem_optimality_of_FP}} \label{app:proof_lemm_parametric_problem_optimality_of_FP}
	\begin{proof}
		By introducing the auxiliary variable $\pi$, the problem $\mathcal{P}_l$ can be equivalently expressed as:
		\begin{align}
			\max_{\mathbf{P}} ~ \pi \quad \text{s.t.}~\pi - \frac{\sum_{k \in \mathcal{K}} R_k\left( \mathbf{P} \right)}{P_{\mathrm{N}}\left( \mathbf{P}, \left\{ R_k\left( \mathbf{P} \right)  \right\}\right)} \leq 0. \label{eq:appendix_B_1}
		\end{align}
		Note that $P_{\mathrm{N}}\left( \mathbf{P}, \left\{ R_k\left( \mathbf{P} \right)  \right\}\right) > 0$, and thus \eqref{eq:appendix_B_1} can be rewritten as:
		\begin{subequations}
			\begin{align}
				 & \max_{\mathbf{P}} ~ \pi                                                                                                                                                \\
				 & \text{s.t.}~\pi P_{\mathrm{N}}\left( \mathbf{P}, \left\{ R_k\left( \mathbf{P} \right)  \right\}\right) - \sum_{k \in \mathcal{K}} R_k\left( \mathbf{P} \right) \leq 0.
			\end{align}
		\end{subequations}
		As demonstrated in \cite{schaibleFractionalProgramming1983}, the above problem is equivalent to finding the root of the following nonlinear function:
		\begin{align}
			F\left( \pi \right) = \max_{\mathbf{P}} ~ \sum_{k \in \mathcal{K}} R_k\left( \mathbf{P} \right) - \pi P_{\mathrm{N}}\left( \mathbf{P}, \left\{ R_k\left( \mathbf{P} \right)  \right\}\right),
		\end{align}
		thus, the condition for global optimality is given by:
		\begin{align}
			F\left( \pi^\ast \right) = 0.
		\end{align}
		This completes the proof.
	\end{proof}

	\section{Proof of \lemmref{lemm:EE_lower_bound}} \label{app:proof_lemm_EE_lower_bound}
	\begin{proof}
		The functions $\hat{f}_k \left( \mathbf{P}, \mathbf{P}_0 \right)$ and $\hat{g}_k \left( \mathbf{P}, \mathbf{P}_0 \right)$ represent the first-order Taylor approximations of $f_k \left( \mathbf{P} \right)$ and $g_k \left( \mathbf{P} \right)$ at the point $\mathbf{P}_0$, respectively. According to the properties of concave functions, we have $\hat{f}_k \left( \mathbf{P}, \mathbf{P}_0 \right) \geq f_k \left( \mathbf{P} \right)$ and $\hat{g}_k \left( \mathbf{P}, \mathbf{P}_0 \right) \geq g_k \left( \mathbf{P} \right)$. Combining equations \eqref{eq:rate}, \eqref{eq:hat_rate}, and \eqref{eq:overline_rate}, we obtain $\hat{R} \left( \mathbf{P}, \mathbf{P}_0 \right) \geq R_k \left( \mathbf{P} \right)$ and $\overline{R} \left( \mathbf{P}, \mathbf{P}_0 \right) \leq R_k \left( \mathbf{P} \right)$. Since $R_k \left( \mathbf{P} \right) \geq 0$ and $P_N \left( \mathbf{P} \right) \geq 0$, where $P_N \left( \mathbf{P} \right)$ is an affine function of $R_k \left( \mathbf{P} \right)$, we can easily deduce that $\overline{\mathrm{EE}}\left( \mathbf{P}, \mathbf{P}_0 \right) \leq \mathrm{EE} \left( \mathbf{P} \right)$.
		Furthermore, the equality of $\hat{f}_k \left( \mathbf{P}, \mathbf{P}_0 \right) \geq f_k \left( \mathbf{P} \right)$ and $\hat{g}_k \left( \mathbf{P}, \mathbf{P}_0 \right) \geq g_k \left( \mathbf{P} \right)$ holds if and only if $\mathbf{P} = \mathbf{P}_0$. Therefore, under the same condition, the equality $\overline{\mathrm{EE}}\left( \mathbf{P}, \mathbf{P}_0 \right) = \mathrm{EE} \left( \mathbf{P} \right)$ also holds.
	\end{proof}

	\section{Proof of \theoref{theo:SLM_stationary_convergent}} \label{app:proof_theo_SLM_stationary_convergent}
	\begin{proof}
		We can derived the following results based on \lemmref{lemm:EE_lower_bound}:
		\begin{subequations} \label{eq:EE_lower_bound_property}
			\begin{align}
				 & \overline{\mathrm{EE}}\left( \mathbf{P}, \mathbf{P} \right) = \mathrm{EE}\left( \mathbf{P} \right),~\forall \mathbf{P}, \label{eq:theorem_1_1}         \\
				 & \overline{\mathrm{EE}}\left( \mathbf{P}, \mathbf{P}^{\prime} \right) \leq \mathrm{EE}\left( \mathbf{P} \right),~\forall \mathbf{P}, \mathbf{P}^\prime,
			\end{align}
		\end{subequations}
		where \eqref{eq:theorem_1_1} reflects the consistency condition in the SCA framework, ensuring that the surrogate function equals the original objective when evaluated at the same point. This condition is essential for convergence.
		Furthermore, the following properties can be easily verified based on the characteristics of $\mathrm{EE}\left( \mathbf{P} \right)$:
		\begin{align}
			 & \frac{\partial \overline{\mathrm{EE}}\left( \mathbf{P}, \mathbf{P}^{\prime} \right)}{\partial \mathbf{P}} \bigg|_{\mathbf{P} \to \mathbf{P}^{\prime}} = \frac{\partial \mathrm{EE}\left( \mathbf{P} \right)}{\partial \mathbf{P}} \bigg|_{\mathbf{P} \to \mathbf{P}^{\prime}}, ~ \forall \mathbf{P}, \tag{\ref{eq:EE_lower_bound_property}{c}} \\
			 & \overline{\mathrm{EE}}\left( \mathbf{P},\mathbf{P}^{\prime} \right) \text{ is continuous in } \left( \mathbf{P}, \mathbf{P}^{\prime} \right), \tag{\ref{eq:EE_lower_bound_property}{d}}                                                                                                                                                        \\
			 & \text{Level set of } \mathrm{EE}\left( \mathbf{P} \right) \text{ is compact}. \tag{\ref{eq:EE_lower_bound_property}{e}}
		\end{align}
		According to \cite[Theorem 1]{razaviyaynUnifiedConvergenceAnalysis2013} and \cite[Corollary 1]{razaviyaynUnifiedConvergenceAnalysis2013}, when \cite[Assumption 1]{razaviyaynUnifiedConvergenceAnalysis2013} holds and the level set of $\mathrm{EE}\left( \mathbf{P} \right)$ is compact, \textit{\textbf{Theorem 1}} holds.
	\end{proof}
\end{appendices}

\bibliographystyle{IEEEtran}
\bibliography{/mnt/d/Literature/Library}

\end{document}